\documentclass[11pt]{article}
 
\usepackage[utf8]{inputenc}
\usepackage{authblk}

\usepackage{fullpage}
\usepackage{soul}
\usepackage{tabularx}

\usepackage{amsthm}
\usepackage{amsmath}
\usepackage{amssymb}
\usepackage{amsfonts}
\usepackage{mathtools}
\usepackage{enumitem}
\usepackage[pagebackref]{hyperref}
\usepackage[sort,numbers]{natbib}
\usepackage[dvipsnames,svgnames,x11names]{xcolor}

\usepackage{tikz}

\usepackage[capitalize]{cleveref}

\usepackage{xspace}
\usepackage{multicol}

\usepackage{todonotes}

\newtheorem{theorem}{Theorem}
\newtheorem{lemma}[theorem]{Lemma}

\newtheorem{corollary}[theorem]{Corollary}
\newtheorem{proposition}[theorem]{Proposition}
\newtheorem{claim}[theorem]{Claim}
\theoremstyle{definition}
\newtheorem{definition}[theorem]{Definition}

\newtheorem{remark}{Remark}
\crefname{claim}{Claim}{Claims}

\DeclareMathOperator{\argmax}{arg\,max}

\newenvironment{claimproof}{\noindent\emph{Proof of Claim.}}{\hfill$\diamond$

}

\newcommand{\problemdef}[3]{
	\begin{center}\fbox{
	\begin{minipage}{0.95\textwidth}
		\noindent
		#1
		\vspace{5pt}\\
		\setlength{\tabcolsep}{3pt}
		\begin{tabularx}{\textwidth}{@{}lX@{}}
			\textrm{Input:}     & #2 \\
			\textrm{Question:}  & #3
		\end{tabularx}
	\end{minipage}}
	\end{center}
}

\newcommand{\problemdeftask}[3]{
	\begin{center}\fbox{
	\begin{minipage}{0.95\textwidth}
		\noindent
		#1
		\vspace{5pt}\\
		\setlength{\tabcolsep}{3pt}
		\begin{tabularx}{\textwidth}{@{}lX@{}}
			\textrm{Input:}     & #2 \\
			\textrm{Task:}  & #3
		\end{tabularx}
	\end{minipage}}
	\end{center}
}

\newcommand{\basicproblem}{\textsc{Adaptive Constructive Coalition Manipulation for Knockout Tournaments}\xspace}
\newcommand{\basicproblemabbrv}{ACCM-KT\xspace}

\newcommand{\nonadaptproblem}{\textsc{Constructive Coalition Manipulation for Knockout Tournaments}\xspace}
\newcommand{\nonadaptproblemabbrv}{CCM-KT\xspace}

\newcommand{\generalizedproblem}{\textsc{Adaptive Constructive Coalition Manipulation for Generalized Knockout Tournaments}\xspace}
\newcommand{\generalizedproblemabbrv}{ACCM-GKT\xspace}

\newcommand{\basicbestresponse}{\textsc{Best Response for ACCM-KT}\xspace}
\newcommand{\basicbestresponseabbrv}{\textsc{BR-ACCM-KT}\xspace}

\newcommand{\generalizedbestresponse}{\textsc{Best Response for ACCM-GKT}\xspace}
\newcommand{\generalizedbestresponseabbrv}{\textsc{BR-ACCM-GKT}\xspace}

\title{Adaptive Manipulation for Coalitions in Knockout~Tournaments}

\author[1]{Juhi~Chaudhary\thanks{Supported by the DAE, Government of India, project nr. RTI4001. (email: juhi.chaudhary@tifr.res.in).}}
\author[2]{Hendrik~Molter\thanks{Supported by the ISF, grant nr.~1470/24 and by the European Union's Horizon Europe research and innovation programme under grant agreement 949707. (email: molterh@post.bgu.ac.il).}}
\author[2]{Meirav~Zehavi$^*$\thanks{Supported by the ISF, grant nr.~1470/24 and by the ERC grant nr.~101039913 (PARAPATH). (email: meiravze@bgu.ac.il).}}

\affil[1]{\small School of Technology and Computer Science, Tata Institute of Fundamental Research, Mumbai, 
India.}
\affil[2]{\small Department of Computer Science, Ben-Gurion~University~of~the~Negev, 
Beer-Sheva, 
Israel.}

\date{}

\begin{document}

\maketitle

\begin{abstract}

\emph{Knockout tournaments}, also known as \emph{single-elimination} or \emph{cup tournaments}, are a popular form of sports competitions. In the standard probabilistic setting, for each pairing of players, one of the players wins the game with a certain (a priori known) probability. Due to their competitive nature, tournaments are prone to manipulation. We investigate the computational problem of determining whether, for a given tournament, a coalition has a manipulation strategy that increases the winning probability of a designated player above a given threshold. More precisely, in every round of the tournament, coalition players can strategically decide which games to throw based on the advancement of other players to the current round. We call this setting \emph{adaptive constructive coalition manipulation}. To the best of our knowledge, while coalition manipulation has been studied in the literature, this is the first work to introduce adaptiveness to this context.

We show that the above problem is hard for every complexity class in the polynomial hierarchy. On the algorithmic side, we show that the problem is solvable in polynomial time when the coalition size is a constant. Furthermore, we show that the problem is fixed-parameter tractable when parameterized by the coalition size and the size of a minimum player set that must include at least one player from each non-deterministic game. Lastly, we investigate a generalized setting where the tournament tree can be imbalanced.

\bigskip

\end{abstract}

\section{Introduction}

\emph{Knockout tournaments}, also known as \emph{single-elimination} or \emph{cup tournaments}, constitute a competition format in which contestants are paired up and compete against each other in rounds, with the losers being eliminated after each round. The tournament continues until only one contestant remains, the \emph{winner}. This format is widely favored in sports~\cite{chaudhary2024make,CR11,GMSS12,suksompong2021tournaments,williams_moulin_2016} and finds applications in several other areas, such as elections and decision-making processes~\cite{vu2009complexity,laslier1997tournament,brandt2007pagerank,Tullock80,rosen1985prizes}. More formally, a knockout tournament consists of $n$ players, where we assume that $n$ is a power of two, and a bijective mapping, called a {\em seeding}, of the players to the leaves of a complete binary tree. As long as at least two players remain, every two players mapped to leaves having the same parent in the tree play a match, whose winner is mapped to the common parent; then, the leaves of the tree are deleted, which means that the losers are knocked out of the tournament. When exactly one player remains, it is declared the winner.

Due to their competitive nature, knockout tournaments are highly prone to various forms of manipulation. 
Such manipulations are frequently reported by the media~(see, e.g., \cite{news,news1,news2,manoli2015only,hill2010critical,feltes2013match}), and their vulnerability was shown empirically (see, e.g.,  \cite{stanton2013structure,mattei2016empirical}). The focus in almost all research on manipulation in knockout tournaments in the literature is on {\em constructive manipulation}, where we aim to make a favorite player win the competition.
Here, the most well-studied forms of manipulation include coalition manipulation (discussed below), and tournament fixing where, generally, the manipulation involves selecting the seeding~\cite{gupta2018rigging,gupta2019succinct,kim2017can,aziz2014fixing,zehavi2023tournament,gupta2018winning,vu2009complexity,williams2010fixing}, and bribery where we can flip the outcome of a bounded, or budget-constrained, number of matches~\cite{konicki2019bribery,mattei2015complexity,kim2015fixing}. 

We focus on constructive {\em coalition manipulation} in knockout tournaments. Constructive coalition manipulation is well-studied in various forms of tournaments (including knockout tournaments) as well as in various settings in voting theory, as discussed in \cref{sec:relatedWorks}. Here, generally, a given subset of players is termed a {\em coalition}. Then, the objective is to determine whether the coalition has a ``strategy''---where coalition players intentionally lose some of their games---so that the favorite player will win or will be likely to win. The simplest setting for such manipulation, considered by Russell and Walsh~\citep{russell2009manipulating}, is when the outcome of every potential match is deterministic and is known in advance. I.e., for every two players, we know with certainty who will win a match between them, if neither of them manipulates. Also, we define a strategy as the set of matches that the coalition players should lose, and for matches involving two coalition players, we further specify which player should lose.
Russell and Walsh~\citep{russell2009manipulating} considered the computational problem of determining, in this setting, whether there exists a strategy for the coalition players that makes the favorite player win and showed that it is solvable in polynomial time. 

Mattei et al.~\citep{mattei2015complexity} study coalition manipulation in the more general case where games have probabilistic outcomes. In their model, the coalition players may modify their winning probabilities before the start of the tournament so that the favorite player will have a higher chance of winning the tournament. However, we will make an argument that it is never non-optimal for a coalition player to either stay at their original winning probability or change it to zero, that is, lose intentionally. Hence, the two models may be considered equivalent.
Mattei et al.~\citep{mattei2015complexity} show that the computational problem of determining whether there exists a strategy for the coalition players such that the favorite player wins the tournament with at least some given probability is contained in NP. They leave open whether this problem is NP-hard.


\subsection{Our Model} 
While the simple setting of Russell and Walsh~\citep{russell2009manipulating} discussed above is an important entry point, it is unrealistic in one way and too restrictive in another. Specifically, on the one hand, it is rarely conceivable that the outcomes of all possible matches will be deterministic---i.e., that we can always be {\em completely certain}, in advance, who beats whom. To overcome this, we employ the standard {\em probabilistic model} for knockout tournaments, which is also used by Mattei et al.~\citep{mattei2015complexity}. Here, for every two players, we know some estimate of the probability of one beating the other. Such estimations can often be derived from available statistics regarding past matches; see, e.g.,~\cite{espn,predict,team,538}.

On the other hand, it is completely unnecessary to make all decisions regarding which manipulations to perform before the tournament has even started. Of course, if the tournament is entirely deterministic, then the timing of these decisions is immaterial. However, when probabilities are involved, the situation changes drastically. When we reach a certain round of the tournament, the coalition players observe what is the {\em current} seeding, that is, which players advanced to this round. According to this crucial information, the coalition players decide which matches to lose in the current round. That is, decisions are made on a round-by-round basis, making use of as much available information as possible. In other words, the strategy is {\em adaptive}. To the best of our knowledge, while coalition manipulation has been studied in the literature, this is the first work to introduce adaptiveness to this context.

Thus, we introduce a new computational problem termed \basicproblem (\basicproblemabbrv). Here, the input consists of:
\begin{multicols}{2}
\begin{itemize}
\item $N$, a set of $n$ players, 
\item a seeding of $N$,
\item for every two players, the probability of one beating the other,
\item a coalition $C\subseteq N$,
\item a favorite player $e^*\in N$, and
\item a probability $t\in[0,1]$.
\end{itemize}
\end{multicols}
Then, the objective is to determine whether there exists an {\em adaptive strategy} (formally defined in \cref{sec:prelim}) for the coalition players so that the probability that $e^*$ will be the winner is at least~$t$.
We assume that the non-coalition players do not behave strategically, that is, the outcome probability of all games not involving coalition players is given in the input.

Of course, the answer to the question above does not yield the strategy itself. Moreover, an explicit description of the strategy is of exponential size, since it encodes, for every round and every {\em possible} advancement of players to that round (of which there exist exponentially many), what should the coalition players do? Fortunately, there is no utility in knowing the entire strategy, as almost all of it concerns hypothetical situations that will not occur. Practically, what the coalition players need to know is what to do in a current round, to which, in particular, we know exactly who are the players that advanced. Thus, we also consider the accompanying {\em best response} version of the problem, where the objective is to compute only the part of the strategy that concerns the tournament's first round. Once the second round is about to start---and then we know who advanced to it---we can simply recompute the best response for the remaining tournament, treated as the input tournament, and so on for the later rounds. We term this problem \basicbestresponse (\basicbestresponseabbrv).

In addition to \basicproblemabbrv\ and \basicbestresponseabbrv, we also investigate the generalized versions of these problems where the tournament tree can be imbalanced. That is, the tournament tree is still binary, but the distance between the root to leaves can differ. Here, the players that participate in a round are those mapped to the leaves with the longest distance to the root. A formal definition can be found in \cref{sec:prelim}. We refer to the generalized problems as \generalizedproblem (\generalizedproblemabbrv) and \generalizedbestresponse (\generalizedbestresponseabbrv).
Lastly, one of our results also resolves an open case for the non-adaptive version of \basicproblemabbrv, which we call \nonadaptproblem (\nonadaptproblemabbrv).

\subsection{Our Results and Proof Ideas} In addition to the introduction of our adaptive model and corresponding computational problems, we prove several highly non-trivial technical contributions, described below. We present the hardness results in \cref{sec:hardness} and the algorithmic results in \cref{sec:algo}. Some of the questions left open are discussed in \cref{sec:conclu}.

 \subparagraph{Classical Hardness.} We first establish the classical hardness of our computational problems. Specifically, we prove that:\begin{itemize}
    \item \basicproblemabbrv is hard for each class in the polynomial hierarchy (PH). 
    \item \generalizedproblemabbrv is PSPACE-hard. (As stated below, we also assert containment in PSPACE.)
    \item \basicbestresponseabbrv is NP-hard.
    \item \nonadaptproblemabbrv is NP-hard. (This case was left open by Mattei et al.~\citep{mattei2015complexity}.)
\end{itemize}

The above hardness results are obtained by reductions from \textsc{Quantified Boolean Formula}~\cite{AB09} with different constraints on the quantifiers and the number of alternations.
The main ingredient for the reductions is a small (constant-sized) subtournament, where a coalition player can decide which of the other players wins the subtournament. We use these subtournaments to create \emph{variable gadgets} for the existentially quantified variables and \emph{clause gadgets}. Intuitively, in the former, the coalition player can select an assignment for an existentially quantified variable, and in the latter, the coalition player can select a literal of the clause (we can assume that each clause has three literals). In the simplest case, where there is only one existential quantifier, we arrange the gadgets (and define the winning probabilities) in a way that each winner of a variable gadget has a non-zero probability of reaching the semi-final game, and each winner of a clause gadget has a non-zero probability of reaching the same semi-final game. Then, the clause player beats the variable player with probability one if the clause player corresponds to a literal that involves the variable corresponding to the variable player and the variable player represents a truth assignment to the variable that does \emph{not} satisfy the literal. Otherwise, the variable player beats the clause player with probability one. In the final game, the winner of the variable and the clause player meets $e^*$, who loses against the clause player with probability one and wins against the variable player with probability one. This way, intuitively, the overall winning probability of $e^*$ is one only if the coalition can prevent a clause player from reaching the final game. This can only happen if a satisfying assignment of the \textsc{Satisfiability} instance is ``selected'' in the variable gadgets and satisfied literals are ``selected'' in the clause gadgets.

To obtain hardness for a class in the PH or PSPACE, informally, we have variable gadgets at different rounds of the overall tournament, alternatingly with players corresponding to universally quantified variables, who randomly advance (each with a non-zero probability) such that every possible assignment has a non-zero probability to occur. This allows coalition players in later variable gadgets to react to which players corresponding to universally quantified variables have advanced. As a consequence, the depth of the tournament depends linearly on the number of quantifier alternations. It follows that if the tournament needs to be balanced, it becomes exponentially large in the number of quantifier alternations. Intuitively, this is the reason why we can ``only'' show hardness for the PH classes for \basicproblemabbrv while we can show PSPACE-hardness for \generalizedproblemabbrv.

\subparagraph{Parameterized Algorithms.}  
We derive our algorithmic results from a versatile dynamic programming algorithm for \generalizedproblemabbrv.
The main idea is the following. We consider the subtree composed of the games that coalition players may participate in, and we consider all configurations of opponents they may have in those games. Now, for all potential sets of games that happen in the same round and where a coalition player is involved, we determine which \emph{strategy profile} for the current round maximizes the winning probability of $e^*$. Note that in the final game, this can be computed directly. Now, we perform dynamic programming along the rounds, from the last to the first. For every round, we use the information already computed for the next round to determine the best strategy profile for the coalition and the corresponding winning probability for $e^*$. By looking at the first round, we obtain the overall best winning probability of $e^*$.

We provide several running time analyses for the dynamic program that yield different tractability results. 
As parameters, we consider the coalition size $|C|$ and the size $x$ of a minimum \emph{random game cover}, that is, the size of a minimum player set that must include at least one player from each non-deterministic game (a formal definition is given in \cref{sec:prelim}).
Notice that $|C|$ is the most natural parameter for all of our computational problems---the size of the coalition is indeed expected to be substantially smaller than the number of all of the players in almost any conceivable scenario. This is due to several factors: larger coalitions are harder to establish as not everyone is willing to throw a match for the manipulator. Even from the manipulator's perspective, adding more players often results in diminishing returns, making the coalition less attractive, or it may draw unnecessary attention, increasing the risk of detection.
Next, we remark that $x$ can be useful as a parameter as well, since for a large number of potential matches, one player will be known to be substantially stronger than the other. Even though the result is unlikely to be deterministic, 
we can (somewhat safely) assign the probability of the stronger player winning to be one.

Specifically, we prove that:
\begin{itemize}
\item \generalizedproblemabbrv can be solved in $n^{O(|C|)}$ time.
\item \basicproblemabbrv can be solved in $(|C|+x)^{O(|C|+ x)}\cdot n^{O(1)}$ time.
\item \generalizedproblemabbrv is contained in PSPACE.
\end{itemize}

Furthermore, we show that the dynamic programming algorithm can also be used to compute best responses. Formally, we show the following.
\begin{itemize}
\item \generalizedbestresponseabbrv can be solved in $n^{O(|C|)}$ time and \generalizedbestresponseabbrv can be solved in polynomial space. 
\item \basicbestresponseabbrv can be solved in $(|C|+x)^{O(|C|+ x)}\cdot n^{O(1)}$ time.
\end{itemize}

\subparagraph{Parameterized Hardness.} Lastly, we complement the aforementioned results by proving the following parameterized hardness result.  
\begin{itemize}
    \item \generalizedproblemabbrv is NP-hard and W[1]-hard when parameterized by the coalition size even if $x$ (defined above) equals two.
\end{itemize}

The above hardness result is obtained by a parameterized reduction from \textsc{Multicolored Clique}~\cite{fellows2009multipleinterval}. Similarly to the reductions for the classical hardness results, the main ingredient is a selection gadget, which is a subtournament with one coalition player who can decide who of the other players wins. The main difference is that the gadget is not a balanced tournament, therefore it can have a polynomial number of players. Intuitively, this is also the reason why our hardness result only holds for \generalizedproblemabbrv and not \basicproblemabbrv.
We use one of these gadgets to select a vertex of each color and one of these gadgets to select an edge for each color combination. Note that this way, the number of coalition players is upper-bounded by the number of colors of the \textsc{Multicolored Clique} instance. The gadgets are now arranged in a way that each winner of a variable gadget has a non-zero probability of reaching the semi-final game, and each winner of a clause gadget has a non-zero probability of reaching the same semi-final game. To do this, we use two \emph{randomize gadgets} (one for the selected vertices and one for the selected edges) that each contain one player that is involved in all games with non-deterministic outcomes. This way, we obtain a random game cover of size two. In the semi-final game, the vertex player beats the edge player with probability one if the edge is from a color combination that involves the color of the vertex and the vertex is \emph{not} an endpoint of the edge. Otherwise, the vertex player beats the edge player with probability one. In the final game, the winner of the vertex and the edge player meets $e^*$, who loses against the edge player with probability one and wins against the vertex player with probability one. This way, intuitively, the overall winning probability of $e^*$ is one only if the coalition can prevent an edge player from reaching the final game. This can only happen if vertices that form a clique are selected in the vertex selection gadgets and the edges in that clique are selected in the edge selection gadgets.

\subsection{Additional Related Works}\label{sec:relatedWorks}

\subparagraph{Coalition Manipulation in Knockout Tournaments and Beyond.} 
Russell and Van Beek~\citep{russell2012detecting} investigated the problem of developing automated tools
for detecting coalitions of teams manipulating the winner in both knockout and round-robin competitions. In the realm of double-elimination tournaments (DETs), Stanton and Williams~\citep{stanton2013structure} established that coalition manipulation of DETs can be polynomially computed under certain restrictions. Here, coalitions of players are vulnerable to manipulation by throwing matches, a phenomenon recently observed in Olympic Badminton. Schneider et al.~\citep{schneider2016condorcet} also
discussed manipulation by coalitions through a collusion between several teams.

Shifting the focus to voting, Walsh~\citep{walsh2011hard} explored the computational complexity of assessing manipulation potential in weighted voting systems, where agents may manipulate outcomes by misrepresenting their preferences. Durand~\citep{durand2015towards} posed the question of whether a strategically voting subset could elect a preferred candidate over the truthful voting outcome. Even simulations on empirical data have been done by Durand~\citep{durand2023coalitional}. Yang~\citep{yang2023complexity} investigated coalition manipulation, revealing that the two-stage majoritarian rule is resilient against most control challenges but susceptible to coalition manipulation.

\subparagraph{Other Forms of Manipulation in Knockout Tournaments and Beyond.} 
The standard manipulation practice seen in knockout tournaments is from the point of view of organizers, who fix the seeding in their interest to make sure their favorite player wins. This problem is known as the \textsc{Tournament
Fixing Problem} (TFP). In TFP, a matrix $P$ is a part of the input, where the entry $P_{i,j}$ gives the probability that player $i$ beats player $j$ \cite{vu2009complexity}. The restricted version of TFP, where $P$ has entries from~$\{0,1\}$ only (the deterministic case), has also been studied and shown to be NP-hard by Aziz et al.~\citep{aziz2014fixing}. 
  A variation of TFP, allowing organizers to both arrange seeding and bribe players to decrease their probability of winning against others at a specified cost, provided it stays within a budget, has been studied by Konicki and Williams~\citep{konicki2019bribery}. In their model, the probability matrix~$P$ is either deterministic or $\varepsilon$-monotonic, reflecting player ordering and constraints on winning probabilities.

Manipulation has also been studied across other tournament formats. For example, \citet{saarinen2015probabilistic} proved that manipulating a \emph{round-robin tournament} by controlling the outcomes of a subset of games to exceed a winning probability threshold is $\#P$-hard, even when probabilities are restricted to $0$, $\frac{1}{2}$, or $1$. In the context of \emph{Challenge the Champ tournaments}, \citet{mattei2015complexity} and more recently \citet{chaudhary2024parameterized} investigated the setting where players can be bribed under a budget to lower their winning probability against the initial champ.


\section{Problem Statements and Preliminaries} \label{sec:prelim}

We begin by defining the following model. 
A tournament has a set of $n$ \emph{players} $N=\{e_1,\ldots, e_n\}$, where for simplicity, we assume that $n$ is a power of two, that is, $n=2^r$. 
Let $\mathbb{Q}_{[0,1]}^{n\times n}$ be the set of $n\times n$ matrices over $[0,1]\cap\mathbb{Q}$ (the rational numbers in the interval $[0, 1]$). Let $p(i,j)$ denote the probability that player $e_i$ will defeat player $e_j$. We denote with 
$P_N =[p(i,j)]_{i,j\in[n]\times[n]} \in \mathbb{Q}_{[0,1]}^{n\times n}$ the \emph{probability matrix} for the tournament.
Note that $p(i,j) + p(j,i) = 1$.

We call a permutation vector $s_1\in N^n$, that is, a vector of length $n$ that contains every player in~$N$ exactly once, a \emph{tournament seeding}. Let $r = \log n$ be the number of rounds of the tournament. Then we call a permutation vector $s_k\in N^{2^{(r-k+1)}}\times\{\bot\}^{2^r - 2^{r-k+1}}$ a \emph{seeding for the $k$th round} of the tournament (which has $2^{(r-k+1)}$ players). Note that a seeding for round $k$ is a vector of the same length as the tournament seeding. The $\bot$-symbol indicates that the seed positions are not present.
A tournament seeding determines how to label the $n$ leaves of an ordered complete binary tree with the players.

Given a set of players $N$ and a seeding $s_1$ of the players, we define a \emph{tournament tree} $T$ to be an ordered rooted complete binary tree with $n$ leaves, where leaves are labeled with the players according to the seeding $s_1$. More formally, the $i$th leaf of the tournament tree is labeled with player $e_j$ if and only if $s_1(i)=j$. For a vertex $v$ in $T$ that is not the root of the tree, we call $v'$ the \emph{sibling} of $v$ if $v$ and $v'$ have the same parent vertex in $T$. See \cref{fig1} for an illustration of a tournament tree.

  \begin{figure}[t]
 \centering
    \includegraphics[scale=1]{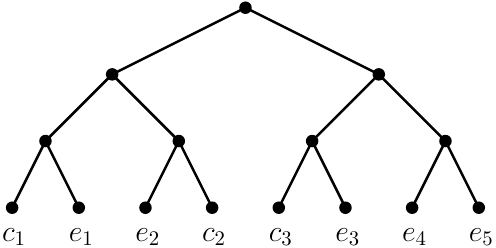}
    \caption{A tournament tree $T$ with 8 players: coalition players $\{c_1,c_2,c_3\}$ and non-coalition players~$\{e_1,\ldots,e_5\}$. The tournament seeding is $(c_1,e_1,e_2,c_2,c_3,e_3,e_4,e_5)$. Here, for example, $c_{1}$ and~$e_{1}$ are siblings of each other.} 
    \label{fig1}
\end{figure}

Given a seeding, the competition is conducted in rounds as follows. As long as the tournament tree has at least two leaves, every two players with a common parent in the tree play against each other, and the winner is promoted to the common parent; then, the leaves of the tree are deleted from it. Eventually, only one player remains, and this player is declared the winner.
Let $s_k$ be a seeding for the $k$th round of the tournament. If in round $k$ of the tournament, the leaves of the binary tree are labeled according to $s_k$, then we say that $s_k$ is obtained in round $k$.

\subsection{Adaptive Coalition Manipulation}

A \emph{coalition} $C\subseteq N$ is a subset of players.
A coalition player can intentionally lose a game, that is, set their winning probability to zero.
A \emph{strategy} for the coalition players is a function $\xi:C\times [\log n]\times (N\cup\{\bot\})^{n}\rightarrow \{0,1\}$ that takes a coalition player, a round number, and a seeding for that round as input, and outputs zero or one. 
Let $e$ be a coalition player, let $k$ be a round number, and let $s_k$ be a seeding for round $k$ of the tournament. If $\xi(e,k,s_k)=0$, then player $e$ intentionally loses their game in round $k$ of the tournament if the seeding $s_k$ is obtained in round~$k$. Otherwise, player $e$ tries to win their game. If two coalition players play against each other, then not both of them can intentionally lose the game.




Formally, we consider the following computational problem.
\problemdef{\basicproblem (\basicproblemabbrv)}{A set of players $N$, a coalition $C\subseteq N$, a favorite player $e^*\in N$, a probability matrix $P_N$, a tournament seeding $s_1$, and a probability threshold $t\in[0,1]$.}{Is there a strategy $\xi$ for the coalition players such that when applied the winning probability of $e^*$ (for the whole tournament) is at least $t$?}

We remark that Mattei et al.~\citep{mattei2015complexity} allow coalition players in their model to lower their probabilities arbitrarily, whereas we only allow coalition players to stay at their original winning probability or lower it to zero. However, observe that at any point during the tournament, we have the following. Fix some coalition player $c$ that is still competing.
\begin{align*}
  p(e^* \text{ wins overall}) = & \ p(e^* \text{ wins overall}\mid c \text{ wins current game})\cdot p(c \text{ wins current game}) \ +\\
  & \ p(e^* \text{ wins overall}\mid c \text{ loses current game}) \cdot p(c \text{ loses current game}).
\end{align*}
It follows that if $p(e^* \text{ wins overall}\mid c \text{ wins current game}) < p(e^* \text{ wins overall}\mid c \text{ loses current game})$, then it is always best for $c$ to intentionally lose the current game, that is, lower the winning probability to zero. Otherwise, it is never non-optimal for the coalition player to stay at its original winning probability. We can conclude that allowing coalition players to lower their winning probabilities arbitrarily does not generalize the model.

\subsection{Best Response}

We remark that a strategy $\xi$ presumably cannot be encoded in polynomial space. As we will show in \cref{thm:PH}, \basicproblemabbrv is presumably not contained in NP. However, in application settings, we might want to output an actual strategy that the coalition players can use rather than only solving the decision problem.
To this end, we introduce the so-called \emph{best response} problem. Intuitively, given a certain round seeding $s_k$ for a tournament round $k$, the best response tells the coalition players which strategy to use in this round to maximize the winning probability of $e^*$. After the first $k$ rounds of the tournament are played out, we can then compute the best response for the next round. This should give $e^*$ the optimal chance to win the tournament. 

To define the best responses formally, we need to know what the best possible winning probability for $e^*$ is in a given tournament. Let $N$ be a set of players (with $e^*\in N$), let $C\subseteq N$ be a coalition, let $P_N$ be a probability matrix, and let $s_1$ be a tournament seeding. We define the \emph{best possible winning probability}
\[
t_\text{opt}=\argmax_{t\in[0,1]} \ \{(N,C,e^*,P_N,s_1,t) \text{ is a yes-instance}\}.
\]
The best possible winning probability $t_\text{opt}^{k}(s_k)$ after round $k$ and a given seeding $s_k$ for the $k$th round is the best possible winning probability of the remaining tournament. Formally, we remove all players that are knocked out from the tournament, and we remove all $\bot$ entries from the seeding for the $k$th round to obtain a new (smaller) instance of \basicproblemabbrv. 

Next, we define a \emph{strategy profile}. It is a function $c:C\rightarrow \{0,1\}$. If $c(e)=0$ for some $e\in C$, we interpret this as $e$ manipulating (that is, intentionally losing) when strategy profile $c$ is used in a certain round of the tournament. A \emph{best response} for a tournament with players $N$, coalition~$C$, probability matrix~$P_N$, and tournament seeding $s_1$ is a strategy profile $c_\text{best}$ with the following property. Let $p(c,s_1,s_2)$ denote the probability that $s_2$ is obtained as a seeding for the second round of the tournament when strategy profile $c$ and tournament seeding $s_1$ is used in the first round. Then the strategy profile $c_\text{best}$ is a best response if we have
\[
\sum_{s_2} p(c_\text{best},s_1,s_2)\cdot t_\text{opt}^{(2)}(s_2) = t_\text{opt}.
\]

Now, we can formulate the computational problem of computing a best response.
\problemdeftask{\basicbestresponse (\basicbestresponseabbrv)}{A set of players $N$, a coalition $C\subseteq N$, a favorite player $e^*\in N$, a probability matrix $P_N$, and a tournament seeding $s_1$.}{Compute a best response $c_\text{best}$ for the first round of the tournament.}

Note that \basicbestresponseabbrv is not formulated as a decision problem. 

\subsection{Non-Adaptive Problem}
In the non-adaptive problem, all coalition players have to specify which games they intend to lose before the tournament starts. We follow the definition of Mattei et al.~\citep{mattei2015complexity}. A \emph{non-adaptive strategy} for the coalition players is a function $\xi:C\times N \rightarrow \{0,1\}$ that takes a coalition player and a player, and outputs zero or one. 
Let $e$ be a coalition player and let $e'$ be another player. If $\xi(e,e')=0$, then player $e$ intentionally loses their game against $e'$ (if it happens in the tournament). Otherwise, player $e$ tries to win the game against $e'$. If two coalition players play against each other, then not both of them can intentionally lose the game.
\problemdef{\nonadaptproblem (\nonadaptproblemabbrv)}{A set of players $N$, a coalition $C\subseteq N$, a favorite player $e^*\in N$, a probability matrix $P_N$, a tournament seeding $s_1$, and a probability threshold $t\in[0,1]$.}{Is there a non-adaptive strategy $\xi$ for the coalition players such that when applied the winning probability of $e^*$ (for the whole tournament) is at least $t$?}

\subsection{Generalized Problem}

Furthermore, we investigate a generalized version of the problem where we do not require the tournament to be balanced. Formally, we allow as a tournament tree every ordered rooted binary tree, that is, an ordered rooted tree where every vertex either has two descendants or is a leaf. This implies that the tournament tree is not solely defined by the player set and the tournament seeding anymore, but also by its structure. 
We define a \emph{generalized tournament tree} $T$ to be an ordered rooted binary tree where the leaves are labeled with a set of players $N$. Let the longest path in~$T$ from the root to a leaf vertex have length $r$. In this case, we say that the tournament has $r$ rounds.
As long as the generalized tournament tree has at least two leaves, every two players that are labels of leaves with maximum distance to the root vertex and with a common parent vertex in the tree play against each other, and the winner is promoted to the common parent; then, the leaves with maximum distance to the root are deleted from the tournament tree. Eventually, only one player remains, and this player is declared the winner.

For some round $k$, we call the vector $s\in(N\cup\{\bot\})^{n}$ containing the labels of the leaves (in some fixed ordering) after $k$ rounds of the tournaments are played (or $\bot$ in case less than $n$ leaves are left) the \emph{configuration of round $k$}. 
In the generalized setting, \emph{strategy} for the coalition players is a function $\xi:C\times [n]\times (N\cup\{\bot\})^{n}\rightarrow \{0,1\}$ that takes a coalition player, a round number, and a configuration for that round as input, and outputs zero or one. Again, the output zero is interpreted as the coalition player intentionally losing the game when facing the corresponding configuration in the corresponding round.

Formally, we define the problem as follows.
\problemdef{\generalizedproblem (\generalizedproblemabbrv)}{A generalized tournament tree $T$ with player set~$N$, a coalition $C\subseteq N$, a favorite player $e^*\in N$, a probability matrix $P_N$, and a probability threshold $t\in[0,1]$.}{Is there a strategy $\xi$ for the coalition players such that when applied the winning probability of $e^*$ (for the whole tournament) is at least $t$?}

We can, in an analogous way, also define a best response for the generalized setting. This leads to the problem \generalizedbestresponse (\generalizedbestresponseabbrv). 

\subsection{Classic and Parameterized Complexity}
We use the standard concepts and notations from classical complexity theory~\cite{AB09}. The complexity class PSPACE contains all problems for which there exists an algorithm that solves all instances in polynomial space. The complexity classes $\Sigma^P_k$ and $\Pi^P_k$ are inductively defined as follows. As the base case we have $\Sigma^P_0=\text{NP}$ and $\Pi^P_0=\text{coNP}$. For $k>0$ we have $\Sigma^P_k=\text{NP}^{\Sigma^P_{k-1}}$ and $\Pi^P_k=\text{coNP}^{\Sigma^P_{k-1}}$.
The polynomial hierarchy (PH) is the union of all complexity classes $\Sigma^P_k$ and $\Pi^P_k$, that is,
$\text{PH}=\bigcup_{k\in\mathbb{N}}(\Sigma^P_k\cup\Pi^P_k)$.

 We use the standard concepts and notations from parameterized complexity theory~\cite{downey2013fundamentals,CyganFKLMPPS15}.
A \emph{parameterized problem} $L\subseteq \Sigma^{*} \times \mathbb{N}$ is a subset of all instances $(x, k)$ from $\Sigma^{*} \times \mathbb{N}$, where $k$ denotes the \emph{parameter}. A parameterized problem $L$ is in the complexity class XP if there is an algorithm that solves each instance $(x,k)$ of $L$
 in $x^{f(k)}$ time, for some computable function $f$. Furthermore,~$L$ is in the class FPT (or fixed-parameter tractable), if there is an algorithm that decides every instance $(x, k)$ for $L$ in $f(k)\cdot|x|^{
O(1)}$ time, where $f$ is any
computable function that depends only on the parameter. If a
parameterized problem $L$ is W[1]-hard, then it is presumably
not fixed-parameter tractable.

\subsection{Random Game Covers}
In this section, we define a minimum \emph{random game cover} and explain how to compute it efficiently. Intuitively, a random game cover is a set of players such that for every game with a non-deterministic outcome, at least one player of the random game cover is involved. If the random game cover is empty, then all games are deterministic and \basicproblemabbrv and \generalizedproblemabbrv can be solved in polynomial time~\cite{russell2009manipulating}. Hence, the size of a minimum random game cover is a natural distance-from-triviality measure for our problems, and we will consider it as a parameterization to obtain tractability results. Formally, a (minimum) random game cover is defined as follows.

\begin{definition}\label{def:rgc}
A \emph{random game cover} for a tournament is a set $X\subseteq N$ of players such that for all pairs of players $i,j$ such that $0\neq p(i,j)\neq 1$ we have that $\{i,j\}\cap X\neq \emptyset$, that is, $i\in X$ or $j\in Y$. A \emph{minimum random game cover} is a random game cover of minimal cardinality.
\end{definition}

We can observe that computing a minimum random game cover is equivalent to computing a minimum vertex cover of the graph with the set of players as vertices and where two players $i,j$ are adjacent if and only if $0\neq p(i,j)\neq 1$. It is well-known that computing a minimum vertex cover can be done in (single exponential) FPT-time with respect to the size of a minimum vertex cover~\cite{CyganFKLMPPS15}. Hence, we get the following.

\begin{proposition}\label{prop:rgc}
    A minimum random game cover can be computed in $2^{O(x)}\cdot n^{O(1)}$ time, where $x$ is the size of a minimum random game cover.
\end{proposition}

\section{Hardness Results}\label{sec:hardness}

In this section, we present our (parameterized) computational hardness results. We first show that \basicproblemabbrv is hard for every complexity class in the polynomial hierarchy (PH). This implies that there presumably is no polynomial-sized witness for yes-instances of the problem. Hence it is unlikely, that a strategy for the coalition players can be encoded in polynomial space. We further show that the generalized problem \generalizedproblemabbrv is PSPACE-complete.
Moreover, we show that \generalizedproblemabbrv is NP-hard and W[1]-hard when parameterized by the coalition size, even if the tournament has a random game cover of size two.
Finally, we show that \basicbestresponseabbrv and \nonadaptproblemabbrv are NP-hard.

We remark that in all our hardness reductions, the entries in the probability matrix of the produced instances are either $0$, $1$, or $\frac{1}{2}$, and we ask whether $e^*$ can win the overall tournament with probability $1$.

\subsection{Classical Hardness Results}

To show that \basicproblemabbrv is hard for each class of the PH, we present a reduction from \textsc{Quantified Boolean Formula} with a constant but arbitrary number of quantifier alternations~\cite{AB09,Sto76}.

\begin{theorem}\label{thm:PH}
\basicproblemabbrv is hard for each class in the PH even if the probability threshold $t$ equals one.
\end{theorem}
\begin{proof}
Fix some $k\in \mathbb{N}$. We present a polynomial-time reduction from \textsc{Quantified Boolean Formula} with $k$ quantifier alternations to \basicproblemabbrv. This shows that \basicproblemabbrv is $\Sigma^P_k$ or $\Pi^P_k$-hard, depending on whether the first quantifier is $\exists$ or $\forall$~\cite{AB09,Sto76}.

Assume we are given a quantified Boolean formula 
\begin{align*}
\exists x_{1,1},\ldots,x_{1,n} \ \forall y_{2,1},\ldots,y_{2,n} \ \exists x_{3,1},\ldots,x_{3,n} \ \ldots \ \forall y_{k,1},\ldots,y_{k,n} \\
\phi(x_{1,1},\ldots,x_{1,n},y_{2,1},\ldots,y_{2,n},x_{3,2},\ldots,x_{3,n},\ldots,y_{k,1},\ldots,y_{k,n}).
\end{align*}

Note that here we assume w.l.o.g.\ that for each quantifier, we have $n$ variables and that $\phi$ is given in conjunctive normal form with clauses of size exactly three. 
We use the letter $x$ to denote existentially quantified variables and the letter $y$ to denote universally quantified variables. The subscript $i,j$ indicates that it is the $j$th variable quantified by the $i$th quantifier.  

Furthermore, for the construction described in the reduction, we assume that the first quantifier of the formula is $\exists$ and that the number of quantifier alternations is odd. Note that if the first quantifier is $\forall$, then we can add a set of existentially quantified variables as the first ones, that do not appear in the formula. Similarly, if the number of quantifier alternations is even (and the first quantifier is existential), then we can add a set of universally quantified variables as the last set of quantified variables, that do not appear in the formula. Note that this only increases the size of the formula by at most a linear factor, and it increases the number of quantifier alternations by at most two.

We construct an instance of \basicproblemabbrv as follows. We start by describing three simple gadgets: one for existentially quantified variables, another for universally quantified variables, and a third for clauses. 

\subparagraph{Existential Variable Gadget.} Let $x_{i,j}$ be an existentially quantified variable. Then, we create four players:
\begin{itemize}
\item A player $x^T_{i,j}$, representing that variable $x_{i,j}$ is set to true.
\item A player $x^F_{i,j}$, representing that variable $x_{i,j}$ is set to false.
\item A player $q_{i,j}$. This player is in the coalition and, informally, will be able to decide the truth assignment of the variable.
\item A player $d_{i,j}$, which is a dummy player that we need for technical reasons.
\end{itemize}

 \begin{figure}[t]
 \centering
    \includegraphics[scale=1]{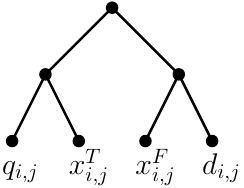}
    \caption{An Existential Variable Gadget $T^{\exists}_{i,j}$.} 
    \label{fig4}
\end{figure}

 We call $x^T_{i,j}$ and $x^F_{i,j}$ ``variable players'', and we set
\begin{multicols}{2}
\begin{itemize}
    \item $p(q_{i,j},x^T_{i,j})=1$,
    \item $p(x^F_{i,j},d_{i,j})=1$,
    \item $p(x^T_{i,j},x^F_{i,j})=1$, and
    \item $p(x^F_{i,j},q_{i,j})=1$.
\end{itemize}
\end{multicols}
The remaining win probabilities are set arbitrarily, say to~$\frac{1}{2}$.  

Now we create a tournament $T^{\exists}_{i,j}$ with the four players and seed $(q_{i,j},x^T_{i,j},x^F_{i,j},d_{i,j})$. See \cref{fig4} for an illustration of the construction of an existential variable gadget.
\begin{claim}\label{claim:existential}
    Either player $x^T_{i,j}$ or player $x^F_{i,j}$ wins the tournament $T^{\exists}_{i,j}$, and the coalition player $q_{i,j}$ can decide which one. 
\end{claim}
\begin{claimproof}
    In any case, $x^F_{i,j}$ beats $d_{i,j}$ in the first round and advances. If $q_{i,j}$ beats $x^T_{i,j}$ in the first round, then $x^F_{i,j}$ beats $q_{i,j}$ in the finals and wins the tournament $T^{\exists}_{i,j}$. If $q_{i,j}$ lets $x^T_{i,j}$ win in the first round, then $x^T_{i,j}$ beats $x^F_{i,j}$ in the finals and wins the tournament.
\end{claimproof}

\subparagraph{Universal Variable Gadget.} Let $y_{i,j}$ be an universally quantified variable. Then, we create two players:
\begin{itemize}
\item A player $y^T_{i,j}$, representing that variable $y_{i,j}$ is set to true.
\item A player $y^F_{i,j}$, representing that variable $y_{i,j}$ is set to false.
\end{itemize}

We also call $y^T_{i,j}$ and $y^F_{i,j}$ ``variable players''. We set $p(y^T_{i,j},y^F_{i,j})=\frac{1}{2}$. Now we create a tournament $T^{\forall}_{i,j}$ with the two players and seed $(y^T_{i,j},y^F_{i,j})$. 

\begin{claim}\label{claim:universal}
    Players $y^T_{i,j}$ and $y^F_{i,j}$ win $T^{\forall}_{i,j}$ with probability $\frac{1}{2}$. 
\end{claim}

\subparagraph{Clause Gadget.} Let $c$ be a clause. Then, we create 8 players:
\begin{itemize}
\item Three players $c_1, c_2, c_3$, representing the three literals of the clause.
\item A player $q_c$. This player is in the coalition and, informally, will be able to select a literal of the clause.
\item Four dummy players $d_{c,1},d_{c,2},d_{c,3},d_{c,4}$ which we need for technical reasons.
\end{itemize}

 \begin{figure}[t]
 \centering
    \includegraphics[scale=1]{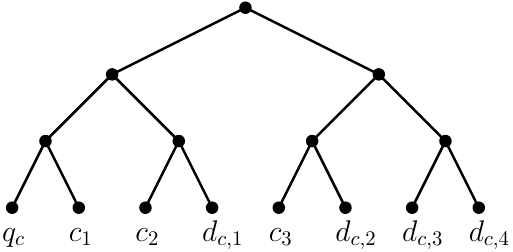}
    \caption{A Clause Gadget $T_{c}$ corresponding to a clause $c$.} 
    \label{fig3}
\end{figure}

 We call $c_1$, $c_2$, and $c_3$ ``clause players'', and we set
\begin{multicols}{2}
\begin{itemize}
    \item $p(q_c,c_1)=1$,
    \item $p(c_2,d_{c,1})=1$,
    \item $p(c_3,d_{c,2})=1$,
    \item $p(d_{c,3},d_{c,4})=1$,
    \item $p(c_1,c_2)=1$,
    \item $p(q_c,c_2)=1$,
    \item $p(c_3,d_{c,3})=1$,
    \item $p(c_1,c_3)=1$, 
    \item $p(c_2,c_3)=1$, and
    \item $p(c_3,q_c)=1$.
\end{itemize}
\end{multicols}

The remaining win probabilities are set arbitrarily, say to $\frac{1}{2}$. 

Now we create a tournament $T_c$ with the eight players and seed $(q_c,c_1,c_2,d_{c,1},c_3,d_{c,2},d_{c,3},d_{c,4})$. See \cref{fig3} for an illustration of the construction of a clause gadget. 
\begin{claim}\label{claim:clause}
    One of the players $c_1,c_2,c_3$ wins the tournament  $T_c$, and coalition player $q_{c}$ can decide which one.
\end{claim}
\begin{claimproof}
 The idea is essentially the same as in the existential variable gadget, and the proof is similar to the one of \cref{claim:existential}.
    In any case, $c_2$ beats $d_{c,1}$ in the first round and advances, and~$c_3$ beats $d_{c,2}$ in the first round and advances. Furthermore, $d_{c,3}$ beats $d_{c,4}$ in the first round and advances.
    In the second round, $c_3$ beats $d_{c,4}$ and advances.
    Note that $q_c$ can beat $c_1$ and $c_2$ with probability one but loses against $c_3$ with probability one.
    If $q_{c}$ lets $c_1$ win in the first round, then $c_1$ goes on to beat $c_2$ in the second round and $c_3$ in the final round and wins the tournament. 
    If $q_{c}$ beats $c_1$ in the first round and lets $c_2$ win in the second round, then $c_2$ goes on to beat $c_3$ in the final round and wins the tournament.
    If $q_{c}$ beats $c_1$ in the first round and beats $c_2$ in the second round, then $q_c$ loses to $c_3$ in the final round, and $c_3$ wins the tournament.
\end{claimproof}



This finishes the description of the gadgets. Furthermore, we introduce our favorite player $e^*$, and we assume that there are sufficiently many dummy players. Now, we set the following winning probabilities.
\begin{itemize}
    \item Variable players beat other variable players with probability $\frac{1}{2}$.
    \item Clause players beat other clause players with probability $\frac{1}{2}$. 
    \item Dummy players beat non-dummy players with probability zero.
    \item A clause player beats a variable player with probability one if the clause player represents a literal involving the variable associated with the variable player, and the truth assignment represented by the variable player does \emph{not} satisfy the literal.

    Otherwise, the variable player beats the clause player with probability one.
    \item Player $e^\star$ beats variable players with probability one and clause players with probability zero.
\end{itemize}

The remaining win probabilities are set arbitrarily, say to $\frac{1}{2}$. We are doing this because when two players whose probabilities have not been previously assigned compete, we ensure that each player has a non-zero probability of winning that match.


Now, we create an existential variable gadget for each existentially quantified variable, a universal variable gadget for each universally quantified variable, and a clause gadget for each clause. In the following, we describe how we arrange these gadgets to create one large tournament.

\subparagraph{Enlarging Gadgets.}
Let $T$ be one of the previously described gadgets and let $t$ denote its height, that is, $t=1$ if $T$ is a universal variable gadget, $t=2$ if $T$ is an existential variable gadget, and $t=3$ if $T$ is a clause gadget. Let $r_1,r_2$ be two positive integers with $r_1+t\le r_2$. Then, intuitively, we want $T(r_1,r_2)$ to be a tournament of height $r_2$ such that the games in $T$ start in round $r_1$ and the winner of $T$ also wins $T(r_1,r_2)$. We do this by introducing sufficiently many dummy players.

Formally, let $N$ denote the set of players in $T$ and let $\sigma:N\rightarrow [2^t]$ denote the seed (note that $|N|=2^t$).
Then first, for every player $i\in N$, we introduce $2^{r_1}-1$ dummy players in $T(r_1, r_2)$. We place those in the seed between player $i$ and player $j$ with $\sigma(j)=\sigma(i)+1$, or to the right of player~$i$ if $\sigma(i)=2^t$.
Next, we introduce $2^{r_2}-2^{r_1+t}$ additional dummy players in $T(r_1, r_2)$ and place them in the seed to the right of player~$i$ with $\sigma(i)=2^t$. This leads to $T(r_1, r_2)$ having a total of $2^{r_2}$ players and consequently a height of $r_2$.  Formally, we have the following.

\begin{claim}\label{claim:enlarge}
    Let $T$ be an existential, universal, or clause gadget. For every strategy of the coalition and every $r_1,r_2$, we have that player $e$ wins $T$ if and only if $e$ wins $T(r_1,r_2)$. Furthermore, if $T$ is an existential or a universal gadget, then the winner of $T(r_1,r_2)$ is determined in round $r_1+3$.
\end{claim}
\begin{claimproof}
Let $N$ denote the set of players in $T$ and let $t$ denote the number of rounds of~$T$. By construction of the enlarging gadget, we have the following: Assume player $i$ and player $j$, with $i,j\in N$, play in round $r$ of $T$, and both of them are not dummy players. Then now, they play in round $r_1+r$ as in the rounds preceding the first round of $T$, they only play against newly added dummy players. It follows that the player $e$ who wins $T$ advances to round $r_1+t$. Moreover, after round $r_1+t$, player $e$ only encounters dummy players. It follows that $e$ also wins $T(r_1,r_2)$. Furthermore, we have that if $T$ is an existential or universal gadget, then $t\le 3$. It follows that in this case, the winner of $T$ (and hence $T(r_1,r_2)$) is determined in round $r_1+3$.
%
\end{claimproof}

Now, we create enlarged variable gadgets and enlarged clause gadgets. 
\begin{itemize}
    \item Let $x_{i,j}$ be an existentially quantified variable, then we create a subtournament $T^{\exists}_{i,j}(3i-3,3k+3)$.
    \item Let $y_{i,j}$ be a universally quantified variable, then we create a subtournament $T^{\forall}_{i,j}(3i-3,3k+3)$.
    \item Let $c$ be a clause, then we create a subtournament $T_{c}(3k,3k+3)$.
\end{itemize}

Let $m$ denote the number of clauses. Note that in total, we have $m+kn$ subtournaments, each of height $3k+3$, that is, each subtournament has $2^{3k+3}$ players. Let $\ell$ be the smallest integer such that $2^{3k+3}(m+kn)\le 2^{\ell}$. Denote $n^\star=2^\ell$. 

Now, we arrange the subtournaments corresponding to enlarged variable gadgets next to each other and add $n^\star-2^{3k+3}kn$ dummy players. Let $T_\text{variables}$ denote the resulting subtournament. Note that $T_\text{variables}$ has $n^\star$ players and height $\ell$. 
We also arrange the subtournaments corresponding to enlarged clause gadgets next to each other and add $n^\star-2^{3k+3}m$ dummy players. Let $T_\text{clauses}$ denote the resulting subtournament. Note that $T_\text{clauses}$ also has $n^\star$ players and height $\ell$ (same as~$T_\text{variables}$).

Finally, we create the overall tournament $T^\star$ by putting $T_\text{variables}$ and 
$T_\text{clauses}$ next to each other, such that the respective winners play against each other, and then we add $e^*$ and $2n^\star-1$ dummy players. See \cref{fig5} for an illustration. 
For all players $a,b$ for which we have not yet specified winning probabilities, we set $p(a,b)=p(b,a)=\frac{1}{2}$.
Note that $T^\star$ has $4n^\star$ players. Hence, the number of players is in $2^{O(k)}(n+m)$. We ask whether there exists a strategy for the coalition players such that $e^*$ wins the tournament with probability one.

 \begin{figure*}[t]
 \centering
    \includegraphics[scale=0.8]{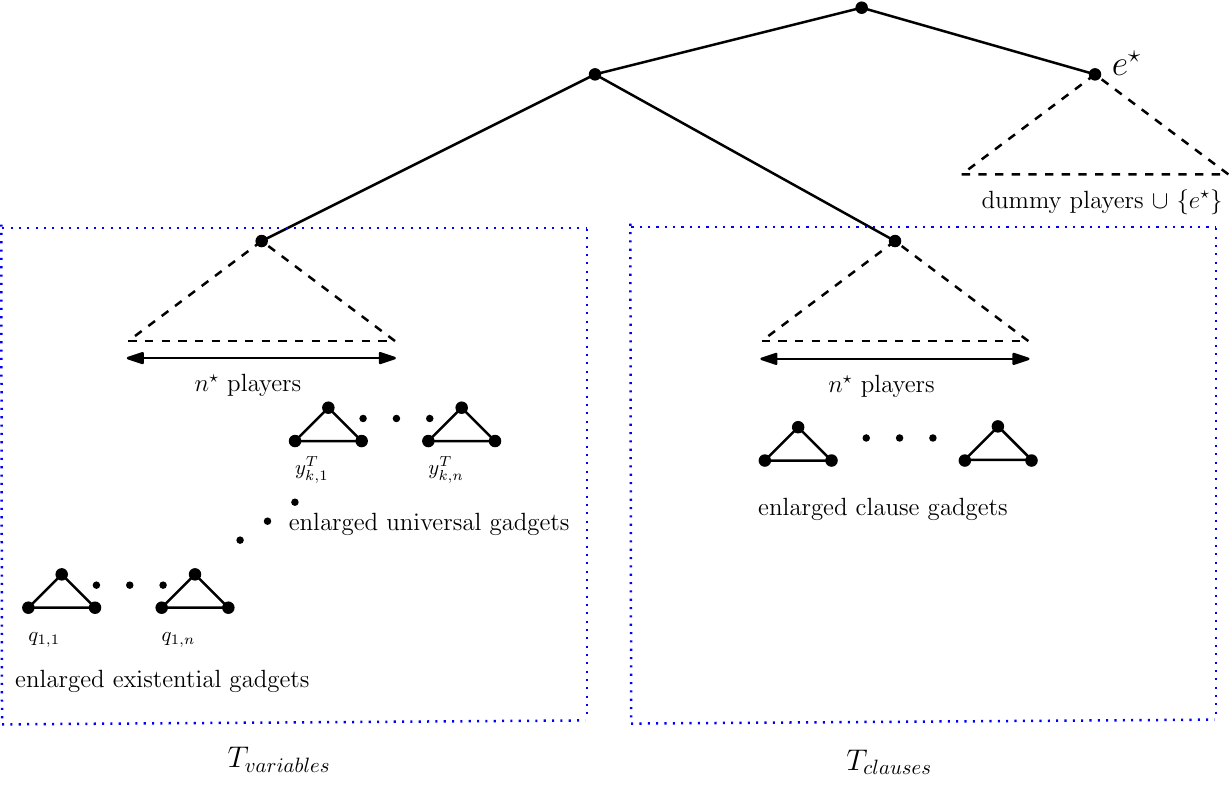}
    \caption{A schematic depiction of different types of gadgets, including enlarged versions of existential gadgets, universal gadgets, clause gadgets, dummy players, and $e^{\star}$, arranged together in tournament $T^{\star}$. Here, only the initial set of existential gadgets and the final set of universal gadgets, which alternate in the arrangement, are shown.} 
    \label{fig5}
\end{figure*}

This finishes the construction, which can clearly be computed in polynomial time for every constant $k$. 
\smallskip

\noindent\emph{Correctness}. In the remainder, we show that the constructed instance is a yes-instance if and only if the \textsc{Quantified Boolean Formula} instance is true. 

\smallskip

$(\Rightarrow)$: Assume that the \textsc{Quantified Boolean Formula} instance is a yes-instance. We create a strategy for the coalition players that leads $e^*$ to win the constructed tournament with probability one as follows. Recall that we assumed that the number $k$ of quantifier alternations is odd and the first quantifier is existential.

Consider the first batch of existentially quantified variables and let $x_{1,j}$ be one of those variables. 
Consider the existential gadget $T^\exists_{1,j}$ which is embedded in the subtournament $T^\exists_{1,j}(0,3k+3)$. By \cref{claim:enlarge}, we know that the player winning $T^\exists_{1,j}(0,3k+3)$ is the player winning $T^\exists_{1,j}$, and by \cref{claim:existential} we know that $q_{1,j}$ can decide which player wins $T^\exists_{1,j}$.
If in the solution for the \textsc{Quantified Boolean Formula} instance variable $x_{1,j}$ is set to true, we choose a strategy for $q_{1,j}$ that ensures that variable player $x^T_{1,j}$ wins $T^\exists_{1,j}$, otherwise we choose a strategy for $q_{1,j}$ that ensures that variable player $x^F_{1,j}$ wins $T^\exists_{1,j}$.

Now consider the second batch of existentially quantified variables and let $x_{3,j}$ be one of those variables. 
Consider the existential gadget $T^\exists_{3,j}$ which is embedded in the subtournament $T^\exists_{3,j}(6,3k+3)$.
Note that the universal gadgets $T^\forall_{2,j}$ are embedded in the subtournaments $T^\forall_{2,j}(3,3k+3)$, and by \cref{claim:enlarge}, we have that at round 6 (when the subtournament $T^\exists_{3,j}$ embedded in  $T^\exists_{3,j}(6,3k+3)$ starts) it is determined which players win the subtournaments $T^\forall_{2,j}(3,3k+3)$ corresponding to the first batch of universal quantified variables.
We create an assignment for the variables in the first batch of universal quantified variables as follows. If the variable player $y^T_{2,j}$ wins the subtournament $T^\forall_{2,j}(3,3k+3)$, then we set the variable $y_{2,j}$ to true, otherwise we set the variable $y_{2,j}$ to false. In the solution to the \textsc{Quantified Boolean Formula} instance, there exists an assignment for the second batch of existential quantified variables for the produced assignment for the first batch of universally quantified variables. We use strategies for the coalition players in the subtournaments $T^\exists_{3,j}(6,3k+3)$ that are derived from the assignment for the second batch of existentially quantified variables analogously to the first one described above.

We continue this process until we have chosen strategies for all coalition players in subtournaments corresponding to enlarged existential gadgets.

It remains to specify a strategy for the coalition players in the subtournaments $T_{c}(3k,3k+3)$ for the clauses $c$. Note that the subtournaments corresponding to the clause gadgets start in round~$3k$. By construction of the tournament and \cref{claim:enlarge}, we have that in this round, all winners of subtournaments corresponding to the variables (both existentially and universally quantified) are determined. These correspond to a satisfying assignment for all variables in the \textsc{Quantified Boolean Formula} instance, hence, for every clause in the formula, we have that at least one literal is satisfied. Let $c$ be a clause satisfied by its $j$th literal. Then we choose a strategy for $q_c$ that ensures that clause player $c_j$ wins the subtournament $T_{c}(3k,3k+3)$, which we can do by \cref{claim:clause}. 

This finishes the description of the strategy for the coalition players. We now prove that $e^*$ wins the constructed tournament with probability one when the described strategy is applied.

Note that by construction, player $e^*$ wins against all variable players and all dummy players with probability one, and loses against all clause players with probability one. Furthermore, by construction, we have that player $e^*$ only encounters dummy players until the final game of the tournament. Now assume for contradiction that player $e^*$ does not win the tournament with probability one. Then, there exists a clause player $c_j$ for some clause $c$ and $j\in\{1,2,3\}$ that can reach the final game with non-zero probability. By construction of the tournament, player $c_j$ plays against a variable player in the semi-final of the tournament. Furthermore, we have that a clause player beats a variable player with probability one only if the clause player corresponds to a literal that involves the variable corresponding to the variable player, and the variable player represents a truth assignment to the variable that does not satisfy the literal.

However, the fact that $c_j$ won the subtournament $T_{c}(3k,3k+3)$ means that the $j$th literal of clause $c$ is satisfied by the assignment to the formula. This is a contradiction to the assumption that player $c_j$ can beat a variable player.

\smallskip

$(\Leftarrow)$: Assume that there exists a strategy for the coalition players that leads $e^*$ to win the constructed tournament with probability one.
We prove that then the \textsc{Quantified Boolean Formula} instance is a yes-instance. Recall that we assume that the number $k$ of quantifier alternations is odd and the first quantifier is existential.

Consider the first batch of existentially quantified variables and let $x_{1,j}$ be one of those variables. 
Consider the subtournament $T^\exists_{1,j}(0,3k+3)$. If player $x^T_{1,j}$ wins the subtournament $T^\exists_{1,j}(0,3k+3)$, then we set $x_{1,j}$ to true; otherwise, we set it to false.

Now consider some assignment for the variables in the first batch of universally quantified variables. By \cref{claim:universal}, the following happens with a non-zero probability: for all universally quantified variables in the first batch,  if $y_{2,j}$ is set to true, then player $y^T_{2,j}$ wins the subtournament $T^\forall_{2,j}(3,3k+3)$; if $y_{2,j}$ is set to false, then player $y^F_{2,j}$ wins the subtournament $T^\forall_{2,j}(3,3k+3)$. 

By \cref{claim:existential,claim:enlarge}, it follows that for this scenario, the strategy for the coalition players determines the winners of the subtournaments $T^\exists_{3,j}(6,3k+3)$. If player $x^T_{3,j}$ wins the subtournament $T^\exists_{3,j}(6,3k+3)$, we set $x_{3,j}$ to true for the given assignment of the first batch of universally quantified variables, otherwise, we set $x_{3,j}$ to false. Now, we continue analogously with the remaining variables. This way, we can produce all possible assignments that have to satisfy the formula of the \textsc{Quantified Boolean Formula} instance.

Assume for contradiction that one of the assignments produced this way does not satisfy the formula of the \textsc{Quantified Boolean Formula} instance. Then, there exists a clause $c$ that is not satisfied by the assignment. By \cref{claim:clause,claim:enlarge}, one of the players $c_1,c_2,c_3$ wins the subtournament $T_c(3k,3k+3)$. Assume w.l.o.g.\ that this player is $c_1$ (the cases where it is $c_2$ or $c_3$ are analogous) and let the existentially quantified variable $x_{i,j}$ be the variable appearing in the first literal of clause $c$ (the case where it is a universally quantified variable is analogous). Assume w.l.o.g.\ that the variable appears non-negated (the case where it appears negated is analogous).
Then we have that in the assignment $x_{i,j}$ is set to false, which means that $x^F_{i,j}$ won the subtournament $T^\exists_{i,j}(3i-3,3k+3)$.

Then, by construction, the following happens with non-zero probability:
Clause player $c_1$ wins subtournament $T_{\text{clauses}}$ and variable player $x^F_{i,j}$ wins the subtournament $T_{\text{variables}}$. Then, by construction, $x^F_{i,j}$ beats $c_1$ in the semi-finals of the tournament, and $x^F_{i,j}$ advances to the finals. Then, in the final match, player $x^F_{i,j}$ beats player $e^*$. This is a contradiction to the assumption that $e^*$ wins the tournament with probability one.
\end{proof}

We remark that the reduction in the proof of \cref{thm:PH} also produces a polynomial-sized instance of \basicproblemabbrv if the number of quantifier alternations in the \textsc{Quantified Boolean Formula} instance is polylogarithmic. Hence, we can get a stronger computational hardness result.
\begin{remark}
\basicproblemabbrv is hard for a complexity class that has \textsc{Quantified Boolean Formula} with polylogarithmically many quantifier alternations as its canonical complete problem. This complexity class contains PH and is contained in PSPACE.
\end{remark}
However, for polynomially many quantifier alternations, the produced instance becomes exponentially large. This is mainly due to the requirement that the tournament tree is a complete (balanced) binary tree. In the general problem \generalizedproblemabbrv, we do not have this requirement, and instances are allowed to have imbalanced tournament trees. Hence, we can straightforwardly modify the reduction in the proof of \cref{thm:PH} to produce instances of \generalizedproblemabbrv that only have a polynomial number of additional dummy players. It follows that the produced instances have polynomial size, even if the \textsc{Quantified Boolean Formula} instance has polynomially many quantifier alternations. Since \textsc{Quantified Boolean Formula} with polynomially many quantifier alternations is PSPACE-hard~\cite{AB09}, we obtain the following.

\begin{theorem}
    \generalizedproblemabbrv is PSPACE-complete even if the probability threshold $t$ equals one.
\end{theorem}
\begin{proof}
We defer proving the containment of \generalizedproblemabbrv in PSPACE to \cref{sec:algo} (\cref{cor:pspace}), where we present our algorithmic results.

For the hardness, we first apply the reduction described in the proof of \cref{thm:PH} to an instance of \textsc{Quantified Boolean Formula} with polynomially many quantifier alternations, which is known to be PSPACE-hard~\cite{AB09}. The size of the obtained instance of \generalizedproblemabbrv is exponential in the size of the \textsc{Quantified Boolean Formula} instance. Hence, we have to reduce its size by removing parts of the tournament tree. We do this in the following way.

We remove every vertex that is \emph{not} a part of a path from the root to a leaf labeled with the name of a non-dummy player. Furthermore, we remove all dummy players from the probability matrix. The resulting tree is not necessarily binary, that is, it may contain vertices that only have one child. To each of those vertices, we add a second child vertex (as a leaf) and label it with the name of a new dummy player. These new dummy players lose against every non-dummy player with a probability of one. All remaining probabilities are set to $\frac{1}{2}$.

The obtained instance now has a polynomial size, since every path from the root to a leaf vertex has a polynomial length. Furthermore, we can compute this instance without explicitly constructing the exponentially large instance obtained from the reduction described in the proof of \cref{thm:PH}. Thus, the instance can be constructed in polynomial time. 

The correctness proof is analogous to the one in the proof of \cref{thm:PH}.
\end{proof}

\subsection{Parameterized Hardness}

In this section, we show that \generalizedproblemabbrv is NP-hard and W[1]-hard when parameterized by the coalition size even on instances with a random game cover of size two. To this end, we present a parameterized reduction from \textsc{Multicolored Clique} parameterized by the number of colors~\cite{fellows2009multipleinterval}.

\begin{theorem}\label{thm:w1}
    \generalizedproblemabbrv is NP-hard and W[1]-hard when parameterized by the coalition size even if the size of a minimum random game cover is two and the probability threshold $t$ equals one.
\end{theorem}
\begin{proof}
     We present a parameterized polynomial-time reduction from \textsc{Multicolored Clique} parameterized by the number of colors, which is known to be W[1]-hard~\cite{fellows2009multipleinterval}. Here, given a $k$-partite graph $G=(V_1\uplus V_2 \uplus\ldots\uplus V_k, E)$, we are asked whether $G$ contains a clique of size $k$. The parameter is~$k$. If $v\in V_i$, then we say that $v$ has \emph{color} $i$. W.l.o.g.\ we assume that $|V_1|=|V_2|=\ldots=|V_k|=n$. Let $E_{i,j}$ denote the set of all edges between vertices from $V_i$ and $V_j$. We assume w.l.o.g.\ that $|E_{i,j}|=m$ for all $i\neq j$. We assume that all vertices and edges are ordered in some fixed but arbitrary way.

     Given an instance of \textsc{Multicolored Clique}, we construct an instance of \generalizedproblemabbrv as follows. The reduction uses two gadgets: a vertex selection gadget and an edge-selection gadget. Both work in a similar way, hence we describe a \emph{generic selection gadget $T_S$} for some set $S$, which is a subtournament, that is, a rooted binary tree, with $|S|+1$ players; one player for each element in $S$ and one additional coalition player. The aim is that the coalition player can ``decide'' which of the players corresponding to the elements in $S$ wins the subtournament.

\subparagraph{Generic Selection Gadget.} Assume we are given a set $S$ and the elements in $S$ are ordered in some fixed but arbitrary way, that is, $S=\{s_1, s_2,\ldots,s_\ell\}$. We create one player for each element in $S$ (we identify those players with their corresponding elements in $S$), and we create one coalition player $c$. We create a tournament for the players $S\cup\{c\}$ as follows. See \cref{fig6} for an illustration.

We create a binary tree $T_S$ with $2\ell+1$ vertices, two vertices $v_i,u_i$ for each element $s_i\in S$, and one vertex $v_c$. Vertex $u_\ell$ is the root of the tree. Each vertex $u_i$ with $i>1$ has vertices $v_i$ and $u_{i-1}$ as children. Vertex $u_1$ has vertices $v_1$ and $v_c$ as children. We can see that all vertices $v_i$ with $i\in[\ell]$ and $v_c$ are leaves. Each vertex $v_i$ is labeled with $s_i\in S$ and vertex $v_c$ is labeled with $c$. 

\begin{figure}[t]
 \centering
    \includegraphics[scale=0.8]{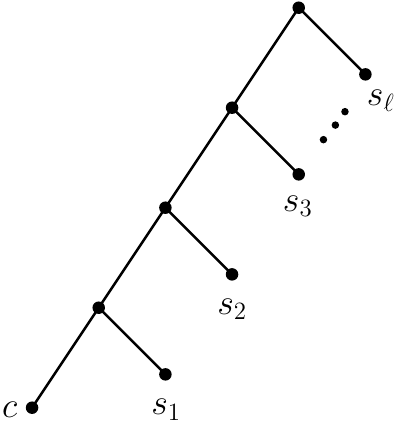}
    \caption{Illustration of a Generic Selection Gadget $T_S$ for some set $S=\{s_1,s_2,\ldots,s_\ell\}$.} 
    \label{fig6}
\end{figure}

Finally, we set the winning probabilities as follows.
\begin{itemize}
    \item For all $1\le i < \ell$ we set $p(v_c,v_i)=1$.
    \item We set $p(v_c,v_\ell)=0$.
    \item For all $1\le i<i'\le \ell$ we set $p(v_i,v_{i'})=1$.
\end{itemize}

This finishes the construction of the gadget. Note that $T_S$ does not contain any random games.
\begin{claim}\label{claim:select}
    One of the players $s_i\in S$ wins the tournament $T_S$, and coalition player $c$ can decide which one.
\end{claim}
\begin{claimproof}
Note that the tournament has $\ell$ rounds, and in each round $1\le i\le \ell$, we have that~$v_i$ is one of the players involved in the game. Furthermore, note that the coalition player $c$ can win against every other player except player $v_\ell$. We show the following: the player against whom the coalition player loses wins the overall tournament. Since $v_\ell$ is the last potential opponent of $c$, player $c$ can choose against which other player to lose. Assume $c$ chooses to lose against player $v_i$. Then $c$ already played and beat players $\{s_{i'}\mid i'<i\}$ and the remaining players in the tournament are $\{s_{i''}\mid i''>i\}$. By construction, player $v_i$ beats all players in $\{s_{i''}\mid i''>i\}$ and hence wins the tournament $T_S$.
\end{claimproof}

Now we create a \emph{vertex selection gadget} $T_{V_i}$ for each color $i\in[k]$ and an \emph{edge selection gadget}~$T_{E_{i,j}}$ for each color combination $1\le i<j\le k$.
We set the following winning probabilities for players from different gadgets.
Let player $s$ correspond to vertex $v$ of color $i$ (this implies that $s$ is a player from the subtournament $T_{V_i}$. Let player $s'$ correspond to some edge $e$. We set $p(s,s')=1$ if $e$ is an edge of a color combination that involves color $i$ and $v$ is \emph{not} one of the endpoints of $e$; otherwise, we set $p(s,s')=0$. 

Next, we describe how to combine the gadgets into the overall tournament. To this end, we introduce \emph{randomize gadgets}.

\subparagraph{Randomize Gadget.} Let $\ell$ be some number. We denote with $T^R_\ell$ a \emph{randomize gadget} with $\ell$ ``slots'', that is, a tournament where $\ell$ of the leaves are labeled with placeholder names. Intuitively, if we label the slots with player names, then each of those players has a non-zero probability of winning the tournament $T^R_\ell$, and no other player in $T^R_\ell$ can win.

  \begin{figure}[t]
 \centering
    \includegraphics[scale=0.7]{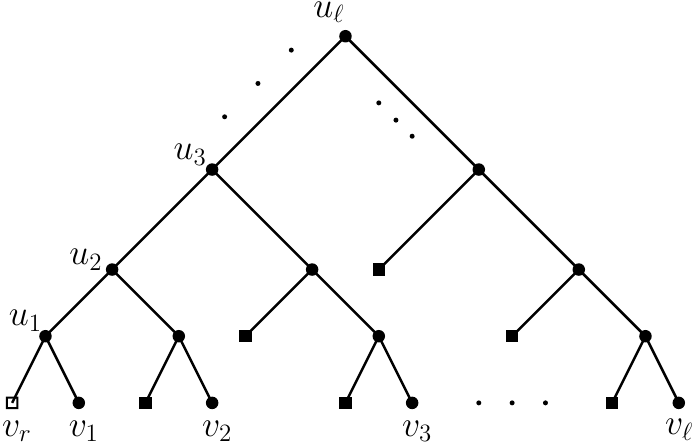}
    \caption{Illustration of the Randomize Gadget. Solid square-shaped leaf vertices are labeled with the names of dummy players. Vertex $v_r$ is labeled with player $r$, which is also the only player in the random game cover of the gadget. Vertices $\{v_1,\ldots,v_\ell\}$ are labeled with placeholder names.} 
    \label{fig8}
\end{figure}

Formally, we create the following binary tree. We create a root vertex $u_\ell$ and vertex $v_r$. We connect $u_\ell$ to $v_r$ with a path of length $\ell$. Let $\{u_1,u_2,\ldots,u_{\ell-1}\}$ denote the internal vertices of that path. Now we create vertices $\{v_1,v_2,\ldots,v_\ell\}$. For all $1\le i\le \ell$, we connect $u_i$ to $v_i$ with a path of length $i$. Now we add an additional leaf vertex to all internal vertices of all paths to vertices $\{v_1,v_2,\ldots,v_\ell\}$ and label each of these leaves with a new dummy player name. We create a player~$r$ and label $v_r$ with $r$. For all $1\le i\le \ell$, we label $v_i$ with a placeholder player name $f_i$. See \cref{fig8} for an illustration.

Finally, we set the winning probabilities as follows.
\begin{itemize}
    \item For all $1\le i < \ell$ we set $p(r,f_i)=\frac{1}{2}$.
    \item We set $p(r,f_\ell)=0$.
    \item For all $1\le i<i'\le \ell$ we set $p(f_i,f_{i'})=1$.
    \item We set the winning probability of a dummy player against a non-dummy player to zero.
    \item We order the dummy players in some fixed but arbitrary way, and for two dummy players $d,d'$ with $d<d'$, we set $p(d,d')=1$.
\end{itemize}

This finishes the construction of the gadget. Note that the size of the gadget is on $O(\ell^2)$ and that~$\{r\}$ is a random game cover for $T^R_\ell$.

\begin{claim}\label{claim:randomselect}
    One of the players $\{f_1,\ldots,f_\ell\}$ wins the tournament $T^R_\ell$, and each player in $\{f_1,\ldots,f_\ell\}$ has a non-zero probability to win $T^R_\ell$.
\end{claim}
\begin{claimproof}
Note that all players $f_i$ with $2\le i\le\ell$ play (and win with probability one) against dummy players until they reach a child vertex of $u_i$. By construction, player $f_i$ wins the tournament if they are paired up against player $r$ in the game corresponding to $u_i$ and win this game. For $i<\ell$, this happens with probability $2^{-i}$, since $r$ has to win $i-1$ games in a row where its winning probability of each game is $\frac{1}{2}$ and finally $f_i$ beats $r$ with probability $\frac{1}{2}$. Finally, player~$f_\ell$ wins with probability $2^{-\ell+1}$ since $r$ has to win $\ell-1$ games in a row where its winning probability of each game is $\frac{1}{2}$ and $f_\ell$ beats $r$ with probability $1$.
\end{claimproof}

We create two randomize gadgets $T_1,T_2$ with $\binom{k}{2}$ slots each. We identify the first $k$ leaves of~$T_1$ that are labeled with placeholder player names with the roots of the vertex selection gadgets~$T_{V_i}$ with $i\in[k]$ and label the remaining leaves with placeholder names with dummy players.
We identify the leaves of $T_2$ that are labeled with placeholder player names with the roots of the edge selection gadgets $T_{E_{i,j}}$ with $1\le i<j\le k$. All players participating in the tournaments rooted at leaves with placeholder names adopt the winning probabilities of the corresponding placeholder players. 
Finally, we introduce player $e^*$. Now we introduce a binary tree $T$ with five vertices $\{v_1,v_2,v_3,v_4,v_5\}$ rooted at $v_1$ and edges such that~$v_1$ has $v_2$ and $v_3$ as children and $v_2$ has $v_4$ and~$v_5$ as children. We identify $v_4$ with the root of $T_1$ and~$v_5$ with the root of $T_2$. We label $v_3$ with~$e^*$. See \cref{fig7_1} for an illustration. Finally, we set the following winning probabilities. Dummy players lose against all non-dummy players with probability one. Player $e^*$ loses against all players~$v$ corresponding to vertices with probability one and wins against all players $e$ corresponding to edges with probability one. 
We order all players in some arbitrary but fixed way, and for all players $a,b$ with $a<b$ for which we have not yet specified a winning probability, we set $p(a,b)=1$.
Finally, we set the probability threshold $t=1$.

\begin{figure}[t]
 \centering
    \includegraphics[scale=1.00]{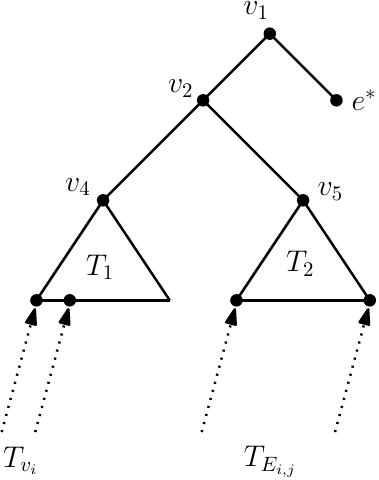}
    \caption{Schematic illustration of the constructed instance of \generalizedproblemabbrv.} 
    \label{fig7_1}
\end{figure}

This finishes the construction, which can clearly be computed in polynomial time. Note that we have $k+\binom{k}{2}$ coalition players and a random game cover of size two (the union of the random game covers of $T_1$ and $T_2$). 

\smallskip

\noindent\emph{Correctness.} In the remainder, we show that the constructed instance of \generalizedproblemabbrv is a yes-instance if and only if the given \textsc{Multicolored Clique} instance is a yes-instance.

\smallskip

$(\Rightarrow)$: Assume the \textsc{Multicolored Clique} instance is a yes-instance and that $X$ is a clique in~$G$. We create a strategy for the coalition players that leads $e^*$ to win the constructed tournament with probability one as follows.

Let $c$ be a coalition player in the subtournament corresponding to the vertex selection gadget for color $i$. Let $\{v\}=X\cap V_i$. By \cref{claim:select}, there is a strategy for $c$ that ensures that the player corresponding to $v$ wins the subtournament $T_{V_i}$. We choose this strategy for $c$.
Similarly, let $c'$ be a coalition player in the subtournament corresponding to the edge selection gadget for color combination $i,j$. Let $e=\{v_i,v_j\}$ be the edge in $E_{i,j}$ with $\{v_i\}=X\cap V_i$ and $\{v_j\}=X\cap V_j$. Since~$X$ is a clique, this edge must exist. By \cref{claim:select}, there is a strategy for $c'$ that ensures that the player corresponding to $e$ wins the subtournament $T_{E_{i,j}}$. We choose this strategy for $c'$.

We claim that this strategy leads $e^*$ to win the tournament with probability one. Assume for contradiction that it does not. Note that by construction, $e^*$ only loses against players corresponding to vertices, and $e^*$ only plays one game---the final game of the tournament. Assume that there is a player $v$ corresponding to a vertex that plays against $e^*$ in the final game with non-zero probability. Let $v$ have color $i$. This means that $v$ won the tournament $T_{V_i}$ and hence $v\in X$. By \cref{claim:randomselect}, we have that in the semi-final game, $v$ plays against a player $e$ corresponding to an edge in $G$. Vertex player $v$ beats edge player $e$ only if the edge connects a vertex $v'$ of color $i$ to some other vertex (say of color $j$) and $v'\neq v$. We also have that $e$ won the subtournament $T_{E_{i,j}}$ and hence $e$ connects two vertices in $X$. However, we have that $v\in X$, and $X$ contains only one vertex of each color. This is a contradiction to $e$ connecting two vertices in $X$.

\smallskip

$(\Leftarrow)$: Assume that there exists a strategy for the coalition players that leads $e^*$ to win the constructed tournament with probability one.
We prove that then the \textsc{Multicolored} instance is a yes-instance.

Let $i$ be a color and let $v$ be the player winning the subtournament $T_{V_i}$. By \cref{claim:select}, we know that $v$ is a player corresponding to a vertex. We say that $v$ is the selected vertex of color $i$. Let $X$ be the set containing the selected vertex of each color. We claim that $X$ is a clique in $G$.

Assume for contradiction that it is not. Then there are two vertices $v,w\in X$ such that $\{v,w\}$ is not an edge in $G$. Let $i$ be the color of $v$ and let $j$ be the color of $w$. Now let $e$ be the player winning the tournament $T_{E_{i,j}}$. By \cref{claim:select}, we know that $e$ is a player corresponding to an edge. Furthermore, we know that $e\neq\{v,w\}$. Assume that $v$ is not an endpoint of $e$ (the case where~$w$ is not an endpoint of $e$ is analogous). By \cref{claim:randomselect}, there is a non-zero probability that~$v$ wins the subtournament~$T_1$, and that $e$ wins the subtournament $T_2$. Furthermore, the strategy of the coalition players cannot depend on the winners of $T_1$ and $T_2$, since by construction, all coalition players are knocked out before the subtournaments  $T_1$ and $T_2$ start.
It follows that $v$ plays against~$e$ in the second-last round of the tournament with non-zero probability. By construction, we have that~$v$ wins with probability one. However, then $v$ plays against $e^*$ in the final game and wins with probability one. Hence, there is a non-zero probability that $e^*$ loses the tournament, a contradiction.
\end{proof}

\subsection{Hardness of Finding Best Responses and Non-Adaptive Strategies}

In this section, we show that computing the best responses is NP-hard and that the non-adaptive problem variant \nonadaptproblemabbrv is NP-hard. Formally, we show the following.

\begin{theorem}\label{thm:resonseNPh}
\basicbestresponseabbrv and \nonadaptproblemabbrv are NP-hard.
\end{theorem}

We remark that the hardness of \basicbestresponseabbrv does not follow immediately from the hardness of \basicproblemabbrv (\cref{thm:PH}): An algorithm computing a best response in polynomial time cannot straightforwardly be used to solve \basicproblemabbrv in polynomial time, since already for the second round there can be exponentially many possibilities on which players advance. However, we can modify the reduction presented in the proof of \cref{thm:PH} in a way that all relevant decisions of the coalition are made in the first round. Hence, the best response is equivalent to a non-adaptive strategy in this case.

\begin{proof}[Proof of \cref{thm:resonseNPh}]
    We present a polynomial-time reduction from the NP-hard problem 3SAT~\cite{Kar72}. The intuition for the reduction is the same as in the reduction presented for \cref{thm:PH}. Given an instance of 3SAT, that is, a Boolean formula $\phi$ in conjunctive normal form and where every clause has size exactly three, we create an instance of \basicbestresponseabbrv and \nonadaptproblemabbrv as follows.

    For each variable, we create an existential variable gadget (see proof of \cref{thm:PH}). Note that in the existential variable gadgets, the coalition player decided in the first round, which of the variable players wins the subtournament of the gadget (see \cref{claim:existential}).
    Now, we introduce a new clause gadget such that the coalition players in the clause gadget make all relevant decisions in the first round.

\subparagraph{New Clause Gadget.} Let $c$ be a clause. Then, we create 8 players:
\begin{itemize}
\item Three players $c_1, c_2, c_3$, representing the three literals of the clause.
\item Three coalition players $q_1, q_2, q_3$. These players, informally, will be able to select a literal of the clause.
\item One \emph{ensurance player} $e$ that, informally, ensures that at least one literal is selected.
\item One dummy player $d$, which we need for technical reasons.
\end{itemize}

We call $c_1$, $c_2$, and $c_3$ ``clause players''.  We set
\begin{itemize}
    \item $p(q_1,c_1)=p(q_2,c_2)=p(q_3,c_3)=1$,
    \item $p(q_1,c_2)=p(q_1,c_3)=0$,
    \item $p(q_2,c_1)=p(q_2,c_3)=0$,
    \item $p(q_3,c_1)=p(q_3,c_2)=0$,
    \item $p(e,q_1)=p(e,q_2)=p(e,q_3)=1$,
    \item $p(e,c_1)=p(e,c_2)=p(e,c_3)=0$, and
    \item $p(e,d)=1$.
\end{itemize}

The remaining win probabilities are set to $\frac{1}{2}$. 

Now we create a tournament $T_c$ with the eight players and seed $(q_1,c_1,q_2,c_2,q_3,c_3,e,d)$. 
\begin{claim}\label{claim:newclause}
    One of the players $c_1,c_2,c_3,e$ wins the tournament $T_c$, and coalition players $q_1,q_2,q_3$ can decide in the first round which one.
\end{claim}
\begin{claimproof}
Note that player $e$ wins against player $d$ in round one with probability one.
In every other pairing, the coalition players can decide whether they want to advance or throw the game.
If all coalition players advance, then $e$ wins the tournament, since $e$ beats each coalition player with probability one.
If coalition player $q_i$ with $i\in[3]$ throws the game in the first round and both other coalition players advance, then $c_i$ wins the tournament, since $c_i$ wins against the other coalition players and against $e$ with probability one.
Lastly, if more than one player from $\{c_1,c_2,c_3\}$ advances to the second round, then each of them has a non-zero probability of winning the tournament, and since every clause player beats $e$ with probability 1, $e$ cannot win the tournament.
\end{claimproof} 

Now, we create a new clause gadget for every clause in $\phi$.  Furthermore, we enlarge the created variable gadgets $T^\exists_x$ to $T^\exists_x(0,3)$ (\cref{claim:enlarge}). Note that for enlarged existential variable gadgets, we have that the winner is determined in round $r_1+1$.
Let $n$ denote the number of variables in $\phi$ and $m$ denote the number of clauses in $\phi$. In total, we have $n+m$ subtournaments, each of height $3$ and with $8$ players. Let $\ell$ be the smallest integer such that $8(n+m)\le 2^\ell$. Denote $n^*=2^\ell$.

Now, we arrange the subtournaments corresponding to enlarged existential variable gadgets next to each other and add $n^\star-8n$ dummy players. Let $T_\text{variables}$ denote the resulting subtournament. Note that $T_\text{variables}$ has $n^\star$ players and height $\ell$. 
We also arrange the subtournaments corresponding to new clause gadgets next to each other and add $n^\star-8m$ dummy players. Let $T_\text{clauses}$ denote the resulting subtournament. Note that $T_\text{clauses}$ also has $n^\star$ players and height $\ell$ (same as~$T_\text{variables}$).

Finally, we create the overall tournament $T^\star$ by putting $T_\text{variables}$ and 
$T_\text{clauses}$ next to each other, such that the respective winners play against each other, and then we add $e^*$ and $2n^\star-1$ dummy players.  
For all players $a,b$ with for which we have not yet specified winning probabilities, we set $p(a,b)=p(b,a)=\frac{1}{2}$.
Note that $T^\star$ has $4n^\star$ players. Hence, the number of players is in $O(n+m)$. 
We ask for a best response for the coalition players that maximized the winning probability of $e^*$.

Since all relevant decisions of the coalition players are made in the first round, we have that a best response constitutes a (full) strategy for the coalition players. It follows that given a best response for the first round, we can compute the winning probability of $e^*$ in polynomial time. We claim that $\phi$ is satisfiable if and only if there is a best response (a non-adaptive strategy for the coalition players) for the created instance of \basicbestresponseabbrv (\nonadaptproblemabbrv) that leads $e^*$ to win the tournament with probability one. The correctness is analogous to the correctness proof of \cref{thm:PH}.
\end{proof}

\section{Algorithmic Results}\label{sec:algo}

In this section, we present our algorithmic results. The main contribution is a dynamic programming algorithm that we analyze in several different ways. From this, we obtain that \generalizedproblemabbrv is in XP when parameterized by the coalition size and that \basicproblemabbrv is fixed-parameter tractable when parameterized by the combination of the coalition size and the size of a minimum random game cover. Furthermore, we get that \generalizedproblemabbrv is contained in PSPACE. Finally, we argue that the presented algorithm can also be used to compute the best responses.

\subsection{The Dynamic Program}

  \begin{figure*}[t]
 \centering
    \includegraphics[scale=1]{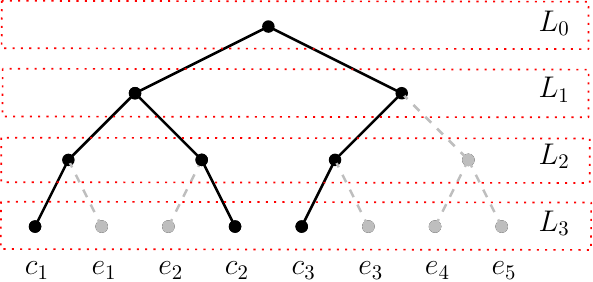}
    \caption{Here, the black edges are the edges of the coalition skeleton $T_{C}$ of $T$, when players $c_{1},c_{2}$, and $c_{3}$ are identified as coalition players. Note that for $L_{1}$, $(c_{1},c_{3})$ and $(c_{2},e_{4})$ are examples of valid configurations while $(c_{3},e_{4})$ is not a valid configuration. Furthermore, for $L_{2}$, both $(e_{4})$ and $(e_{5})$ represent valid sibling configurations.} 
    \label{fig2}
\end{figure*}

We start by describing a dynamic programming algorithm. To this end, we need to introduce some concepts first. 
Given a set of players $\{e_1, e_2, \ldots, e_n\}$ and a coalition $C\subseteq N$%
, we define the \emph{coalition skeleton} $T_C$ to be the subtree of the tournament tree $T$ consisting of all paths from the root of $T$ to a leaf that is labeled with a player in $C$.
We call the $i$th \emph{level}~$L_i$ of the coalition skeleton $T_C$ the set of vertices in $T_C$ that have distance $i$ from the root of $T_C$. See \cref{fig2} for an illustration. Note that every level of $T_C$ contains at most $|C|$ vertices.
We call a subset~$S$ of players a \emph{valid configuration} of the $i$th level $L_i$ of $T_C$ if $|S|=|L_i|$ and for every vertex $v\in L_i$, there is a player $e\in S$ such that the subtree of $T$ rooted at $v$ has a leaf that is labeled with~$e$.
We say that~$v$ is the position of $e$ in~$L_i$.
Intuitively this means, that if player $e$ reaches the $(r-i)$th round of the tournament, then $e$ is mapped to $v$ in the seeding for round $r-i$.
We call a subset $S$ of players a valid \emph{sibling configuration} of the $i$th level~$L_i$ of~$T_C$ if $|S|$ equals the number of vertices in $L_i$ that have a sibling which is not contained in $L_i$ and if
for every vertex $v\in L_i$ that has a sibling $v'\notin L_i$ there is a player in $e\in S$ such that the subtree of $T$ rooted at $v'$ has a leaf that is labelled with $e$. 
We say that $v'$ is the position of $e$ in~$L_i$.
It is crucial to note that valid sibling configurations do not include players from the coalition set $C$.

Let $S$ be a valid configuration for $L_i$ and let $S'$ be a valid sibling configuration for $L_i$. For players $e,e'\in S\cup S'$, we say that players $e$ and $e'$ are siblings if $v$ is the position of $e$ in $L_i$ and $v'$ is the position of $e'$ in $L_i$, and $v$ and $v'$ are siblings in $T$. If players $e$ and $e'$ are siblings we denote $e'=s_{S,S'}(e)$.
The probability $p(S)$ of a sibling configuration $S$ for level $L_i$ is the probability that all players in $S$ advance to the $(r -i)$th round of the tournament, where $r\le n$ is the total number of rounds of the tournament. Note that $p(S)$ can be computed by a straightforward dynamic program.

Given a valid configuration $S$ for $L_i$, a valid sibling configuration $S'$ for $L_i$, a valid configuration~$S^\star$ for $L_{i-1}$, and a strategy profile $c$, we denote with $p(S,S',S^\star,c)$ the probability that configuration $S^\star$ is obtained for $L_{i-1}$ assuming configuration $S$ and sibling configuration $S'$ are present in $L_{i}$ and strategy profile $c$ is used by the coalition players in round $r-i$ of the tournament. Formally, we have
\begin{align*}
  p(S,S',S^\star,c)= & \prod_{e\in S^\star\cap C} c(e)\cdot p(e,s_{S,S'}(e))\cdot \\
   & \prod_{e\in S^\star\setminus C} p(e,s_{S,S'}(e)).  
\end{align*}


Let $\mathcal{S}\subseteq 2^{\{e_1,e_2,\ldots,e_n\}}$ be the set of all possible configurations. We create a dynamic program $M:\{0,1,2,\ldots,r\}\times \mathcal{S}\rightarrow [0,1]$ that maps a level $L_i$ of $T_C$ together with a valid configuration~$S$ for $L_i$ to the probability that $e^\star$ wins assuming the players in $S$ are seeded directly into their positions in $L_i$ and assuming the coalition players maximize the winning probability of  $e^\star$. We define $M$ as follows. We have
\[
M[0,\{e^\star\}]=1,
\]
and for all $S\subseteq\{e_1,e_2,\ldots,e_n\}$ with $S\neq \{e^\star\}$, we have
\[
M[0,S]=0.
\]
Let $S\subseteq\{e_1,e_2,\ldots,e_n\}$ be valid for $L_i$ for $i>0$, then we have 

\begin{equation*}
M[i, S] = \sum_{\substack{\text{valid sibling configuration} \\ \substack{S'} \ \text{for} \ L_{i}}} \left( p(S') \cdot \max_{\substack{\text{strategy profile} \\ \substack{c}}} \left( \sum_{\substack{\text{valid configuration} \\ \substack{S^\star \text{ for} \ L_{i-1}}} } p(S, S', S^\star, c) \cdot M[i-1, S^\star] \right) \right).
\end{equation*}



Before we analyze the running time of the algorithm, we prove that the presented dynamic programming algorithm is correct.

\begin{lemma}\label{lem:DPcorrect}
    For all $i\in \{0,1,2,\ldots,r\}$ and all valid configurations $S\subseteq\{e_1,e_2,\ldots,e_n\}$ for $L_i$ we have that $M[i,S]$ is the probability that $e^\star$ wins assuming the players in $S$ are seeded directly into their positions in $L_i$.
\end{lemma}
\begin{proof}
    We prove the lemma statement by induction on the level $i$. For $i=0$, we have $M[0,\{e^\star\}]=1$, and for all $S\subseteq\{e_1,e_2,\ldots,e_n\}$ with $S\neq \{e^\star\}$, we have $M[0,S]=0$. Hence, for $i=0$, the dynamic program is clearly correct.

    Assume $i>0$. Recall that the configuration $S$ specifies the player positions within $L_i$. In order to compute a configuration for the next round of the tournament, that is, level $i-1$, we need to know which players the players in $S$ play against. This is specified by the sibling configuration $S'$. The probability of attaining sibling configuration $S'$ is represented by $p(S')$. Given a configuration $S$ and a sibling configuration $S'$ for level $i$, the probability of player advancement relies on the strategic choices made by the coalition players. Since the coalition players may choose their strategy, we consider the strategy profile that maximizes  $ e^{*}$'s winning probability. 
    
    Given a configuration $S$, a sibling configuration $S'$, and a strategy profile $c$, the probability of achieving configuration $S^\star$ in level $i-1$ is denoted as $p(S,S',S^\star,c)$. By induction hypothesis, the probability that $e^\star$ wins given that $S^\star$ is the configuration of $L_{i-1}$ is $M[i-1,S^\star]$. 
    
    The above arguments can directly be translated into the recursive formula for $M[i,S]$. Consequently, we can conclude that the lemma statement is true.
\end{proof}

\subsection{Running Time Analysis}

Depending on the way we analyze the running time of the dynamic programming algorithm, we obtain several results. We first show that \generalizedproblemabbrv is contained in XP when parameterized by the coalition size.
\cref{thm:w1} implies that we cannot expect to obtain an FPT-algorithm for this parameterization.

\begin{theorem}\label{thm:xp}
\generalizedproblemabbrv can be solved in $n^{O(|C|)}$ time.
\end{theorem}

\begin{proof}
Observe that there is exactly one valid configuration for level $r$, namely, $C$. By \cref{lem:DPcorrect}, it follows that $M[r,C]$ is the probability that $e^\star$ wins the tournament. 

In the remainder, we show that the time required to compute all entries of $M$ (and hence, in particular, the entry $M[r,C]$) is in $n^{O(|C|)}$.

First, note that the size of the dynamic programming table is in $O(n^{|C|}\cdot n)$, since all valid configurations contain at most $|C|$ players and there are at most $n$ rounds in the tournament. It remains to analyze the time needed to compute one entry of the table. The first sum in the recursive expression sums over all valid sibling configurations. A sibling configuration is a subset of the players of size at most $|C|$. There are $n^{O(|C|)}$ such subsets, and we can check for each one in polynomial time whether it is a valid sibling configuration. Next, we have a maximum over all strategy profiles, of which there are $2^{|C|}$ many. Lastly, we sum over all valid configurations for the previous level. There are $n^{O(|C|)}$ such configurations. Furthermore, the functions $p$ (note that we overloaded notation here) can clearly be computed in polynomial time. It follows that the overall time required to compute all entries of $M$ is in $n^{O(|C|)}$.
\end{proof}

Next, we show how to obtain an FPT-algorithm for \basicproblemabbrv when parameterized by the combination of the coalition size and the size of a minimum random game cover. \cref{thm:w1} implies that we cannot expect to get an FPT-algorithm for \generalizedproblemabbrv with the same parameterization.

\begin{theorem}\label{thm:fpt}
\basicproblemabbrv can be solved in $(|C|+x)^{O(|C|+ x)}\cdot n^{O(1)}$ time, where $x$ is the size of a minimum random game cover.
\end{theorem}

\begin{proof}
Again, as in the proof of \cref{thm:xp}, observe that there is exactly one valid configuration for level $r$, namely, $C$. Furthermore, in the case of \basicproblemabbrv, we have that $r=\log n$. By \cref{lem:DPcorrect}, it follows that $M[\log n,C]$ is the probability that $e^\star$ wins the tournament. 

In the remainder, we show that the time required to compute $M[\log n,C]$ is in $(|C|+x)^{O(|C|+ x)}\cdot n^{O(1)}$.
To this end, we present an alternative way to enumerate all \emph{reachable} (sibling) configurations. More precisely, 
a valid sibling configuration $S'$ is reachable if $p(S')\neq 0$. Initially, we have that~$C$ is a reachable configuration (for level $\log n$) and all valid configurations $S^\star$ for level $i-1$ are reachable if there is a reachable configuration $S$ for level $i$, a reachable sibling configuration $S'$, and a strategy profile $c$ such that $p(S,S',S^\star,c)\neq 0$.

We start by computing a minimum random game cover in $2^{O(x)}\cdot n^{O(1)}$ time (\cref{prop:rgc}). Now we can observe that each overall outcome of the tournament is uniquely defined by the number of wins achieved by each player in the random game cover and each player in the coalition: the outcomes of each game that does not involve a player from the random game cover or a coalition player is uniquely determined. Furthermore, if we can enumerate all possible outcomes of the tournament, we can also enumerate all reachable valid configurations and all reachable valid sibling configurations. Each player in the random game cover can win at most $\log n+1$ games. Hence, there are $O((\log n)^{|C|+x})$ possible outcomes of the tournament and hence also as many valid (sibling) configurations.

Using this, we get that the number of entries of the dynamic programming table that we need to compute in order to compute $M[\log n,C]$ is in $O((\log n)^{|C|+x+1})$. It remains to analyze the time needed to compute one entry of the table. The first sum in the recursive expression sums over all valid sibling configurations. Note that we only need to enumerate all reachable valid sibling configurations to compute the recursion. The number of reachable valid sibling configurations is in $O((\log n)^{|C|+x})$. Next, we have a maximum over all strategy profiles, of which there are $2^{|C|}$ many. Lastly, we sum over all valid configurations for the previous level. Again, it is sufficient to sum over all reachable valid sibling configurations, of which there are $O((\log n)^{|C|+x})$ many. Furthermore, the functions $p$ (note that we overloaded notation here) can clearly be computed in polynomial time. 
Now we use the well-known fact that for $k\le n$ we have $(\log n)^{O(k)}\subseteq k^{O(k)}+n^{O(1)}$ (see e.g.~\cite{RamanujanS17}).
It follows that the overall time required to compute all entries of $M$ is in $(|C|+x)^{O(|C|+ x)}\cdot n^{O(1)}$.
\end{proof}

Lastly, we observe the following: if we use the definition of the dynamic program as a recursive algorithm, that is, we recompute each value in the recursion instead of looking it up in the table, then we can solve \generalizedproblemabbrv using only polynomial space. Hence, we obtain the following.

\begin{theorem}\label{cor:pspace}
    \generalizedproblemabbrv is contained in PSPACE.
\end{theorem}

Finally, we consider the problem of computing best responses. Notice when computing the entry $M[r,C]$ of the dynamic programming table, we have only one valid sibling configuration for $L_{r}$, say $S'$, namely the players that are seeded as the siblings of the coalition players. Hence, we have that a best response for the first round is a strategy profile that maximizes the expression
\begin{equation*}
\max_{\substack{\text{strategy profile} \\ \substack{c}}} \left( \sum_{\substack{\text{valid configuration} \\ \substack{S^\star \text{ for} \ L_{r}}} } p(C, S', S^\star, c) \cdot M[r-1, S^\star] \right).
\end{equation*}

It follows that the dynamic programming algorithm can also compute best responses (in the same running time bounds). Formally, we get the following.
\begin{corollary}
\generalizedbestresponseabbrv can be solved in $n^{O(|C|)}$ time and \generalizedbestresponseabbrv can be solved in polynomial space. \basicbestresponseabbrv can be solved in $(|C|+x)^{O(|C|+ x)}\cdot n^{O(1)}$ time, where $x$ is the size of a minimum random game cover. 
\end{corollary}



\section{Conclusion} \label{sec:conclu}

Our work introduces adaptiveness to the concept of coalition manipulation in probabilistic knockout tournaments, and presents several highly non-trivial results in this regard. Of course, the same questions can be asked for the concept of budget-constrained manipulation in knockout tournaments. Since coalition manipulation is a special case of budget-constrained manipulation (assign an ``infinite'' price to manipulations of matches involving only non-coalition players, and a price of $1$ to manipulations of matches involving at least one coalition player), all of our hardness results extend to this setting as well. However, the study of positive results in this setting is left for future research. Additionally, we would like to pose the following questions for future research:

\begin{itemize}
    \item Is \basicproblemabbrv PSPACE-hard? Recall that we have shown that \basicproblemabbrv is hard for each class in the PH, but not for PSPACE.
    \item Do the non-generalized problems considered in this paper belong to the class FPT when parameterized by the coalition size? Recall that we have shown, for these problems, containment in XP (or in FPT when the size $x$ of a minimum random game cover is also part of the parameterization), and for the generalized problem \generalizedproblemabbrv, W[1]-hardness.
    \item Do the non-generalized problems considered in this paper belong to the class FPT when parameterized by $x$? Recall that we have shown, for these problems, containment in FPT when the coalition size is also part of the parameterization), and for the generalized problem \generalizedproblemabbrv, NP-hardness for $x=2$.
    \item Which other natural parameters of the problems yield tractability? One possibility to derive such a parameter is to consider alternative measures (to the size of a minimum random game cover) for the closeness of the probability matrix to a deterministic one.
\end{itemize}

\bibliographystyle{abbrvnat}
\bibliography{bibliography}	

\end{document}